\documentclass{fundam}

\usepackage[utf8]{inputenc}
\usepackage{amsfonts}
\usepackage{amsmath}
\usepackage{stmaryrd}
\usepackage{tikz}
\usepackage{hyperref}

\title{Stackelberg-Pareto Synthesis with Quantitative Reachability Objectives}

\author{Thomas Brihaye\thanks{Partly supported by the F.R.S.- FNRS under grant n°T.0027.21.}\\
Université de Mons (UMONS), Belgium
\and Véronique Bruyère\thanks{Partly supported by the F.R.S.- FNRS under grant n°T.0023.22.}\\
Université de Mons (UMONS), Belgium 
\and Gaspard Reghem\\
ENS Paris-Saclay, Université Paris-Saclay, France}

\runninghead{Th.  Brihaye, V. Bruyère, and G. Reghem}{Quantitative Reachability Stackelberg-Pareto Synthesis}

\newcommand{\nat}{\mathbb{N}}
\newcommand{\play}{\textsf{Play}}
\newcommand{\hist}{\textsf{Hist}}
\newcommand{\histi}[1]{\textsf{Hist}^{#1}}
\newcommand{\out}{\textsf{out}}

\newcommand{\val}[1]{\textsf{val}(#1)}
\newcommand{\wit}{\textsf{Wit}}
\newcommand{\prefix}{\sqsubseteq}

\NewDocumentCommand{\Last}{m}{\textsf{last}{(#1)}}
\newcommand{\cost}[1]{\textsf{cost}({#1})}
\newcommand{\costT}[2]{\textsf{cost}_{#1}(#2)}
\newcommand{\DimGame}[1]{\textsf{Games}_{#1}}
\newcommand{\length}[1]{\textsf{length}({#1})}
\newcommand{\BinDimGame}[1]{\textsf{BinGames}_{#1}}
\newcommand{\SPS}[2]{\mbox{SPS}({#1},{#2})}
\newcommand{\f}[2]{f({#1},{#2})}

\newcommand{\mem}[1]{\textsf{rec}({#1})}
\newcommand{\punStrat}[1]{\tau^{\textsf{Pun}}_{#1}}

\begin{document}

\maketitle

\begin{abstract}
In this paper, we deepen the study of two-player Stackelberg games played on graphs in which Player~$0$ announces a strategy and Player~$1$, having several objectives, responds rationally by following plays providing him Pareto-optimal payoffs given the strategy of Player~$0$. The Stackelberg-Pareto synthesis problem, asking whether Player~$0$ can announce a strategy which satisfies his objective, whatever the rational response of Player~$1$, has been recently investigated for $\omega$-regular objectives. We solve this problem for weighted graph games and quantitative reachability objectives such that Player~$0$ wants to reach his target set with a total cost less than some given upper bound. We show that it is \textsf{NEXPTIME}-complete, as for Boolean reachability objectives.
\end{abstract}

\section{Introduction}
%=====================

\emph{Formal verification}, and more specifically \emph{model-checking}, is a branch of computer science which offers techniques to check automatically whether a system is correct~\cite{DBLP:reference/mc/BloemCJ18,DBLP:books/daglib/0007403-2}. This is essential for systems responsible for critical tasks like air traffic management or control of nuclear power plants. Much progress has been made in model-checking both theoretically and in tool development, and the technique is now widely used in industry.

Nowadays, it is common to face more complex systems, called \emph{multi-agent systems}, that are composed of heterogeneous components, ranging from traditional pieces of reactive code, to wholly autonomous robots or human users.  Modelling and verifying such systems is a challenging problem that is far from being solved. One possible approach is to rely on \emph{game theory}, a branch of mathematics that studies mathematical models of interaction between agents and the understanding of their decisions assuming that they are \emph{rational} \cite{Osborne94,vonNeumannMorgenstern44}. Typically, each agent (i.e. player) composing the system has his own objectives or preferences, and the way he manages to achieve them is influenced by the behavior of the other agents. 

Rationality can be formalized in several ways. A famous model of agents’  rational behavior is the concept of \emph{Nash equilibrium} (NE)~\cite{Nas50} in a multiplayer non-zero sum game graph that represents the possible interactions between the players~\cite{Ummels08}. Another model is the one of \emph{Stackelberg games}~\cite{Stackelberg37}, in which one designated player -- the leader, announces a strategy to achieve his goal, and the other players -- the followers, respond rationally with an optimal response depending on their goals (e.g. with an NE). This framework is well-suited for the verification of correctness of a controller intending to enforce a given property, while interacting with an environment composed of several agents each having his own objective. In practical applications, a strategy for interacting with the environment is committed before the interaction actually happens. 

\paragraph{Our contribution.}  In this paper, we investigate the recent concept of two-player Stackelberg games, where the environment is composed of \emph{one player} aiming at satisfying \emph{several objectives}, and its related \emph{Stackelberg-Pareto synthesis} (SPS) problem~\cite{DBLP:journals/corr/abs-2203-01285,DBLP:conf/concur/BruyereRT21}. In this framework, for Boolean objectives, given the strategy announced by the leader, the follower responses rationally with a strategy that ensures him a vector of Boolean payoffs that is \emph{Pareto-optimal}, that is, with a maximal number of satisfied objectives. This setting encompasses scenarios where, for instance, several components of the environment can collaborate and agree on trade-offs. The SPS problem is to decide whether the leader can announce a strategy that guarantees him to satisfy his own objective, whatever the rational response of the follower. 

The SPS problem has been solved in~\cite{DBLP:conf/concur/BruyereRT21} for \emph{$\omega$-regular} objectives. We here solve this problem for weighted game graphs and \emph{quantitative reachability} objectives for both players. Given a target of vertices, the goal is to reach this target with a cost as small as possible. In this quantitative context, the follower responds to the strategy of the leader with a strategy that ensures him a Pareto-optimal cost vector given his series of targets. The aim of the leader is to announce a strategy in a way to reach his target with a total cost less than some given upper bound, whatever the rational response of the follower. We show that the SPS problem is \textsf{NEXPTIME}-complete (Theorem~\ref{thm:main}), as for Boolean reachability objectives.

It is well-known that moving from Boolean objectives to quantitative ones allows to model \emph{richer properties}. This paper is a first step in this direction for the SPS problem for two-player Stackelberg games with multiple objectives for the follower. Our proof follows the same pattern as for Boolean reachability~\cite{DBLP:conf/concur/BruyereRT21}: if there is a solution to the SPS problem, then there is one that is finite-memory whose memory size is at most exponential. The non-deterministic algorithm thus guesses such a strategy and checks whether it is a solution. However, a crucial intermediate step is to prove that if there exists a solution, then there exists one whose Pareto-optimal costs for the follower are \emph{exponential} in the size of the instance (Theorem~\ref{thm:bounded-costs}). The proof of this non trivial step (which is meaningless in the Boolean case) is the main contribution of the paper. Given a solution, we first present some hypotheses and techniques that allow to locally modify it into a solution with smaller Pareto-optimal cost vectors. We then conduct a proof by induction on the number of follower's targets, to globally decrease the cost vectors and to get an exponential number of Pareto-optimal cost vectors. The \textsf{NEXPTIME}-hardness of the SPS problem is trivially obtained by reduction from this problem for Boolean reachability. Indeed, the Boolean version is equivalent to the quantitative one with all weights put to zero and with the given upper bound equal to zero. Notice that the two versions differ: we exhibit an example of game that has a solution to the SPS problem for quantitative reachability, but none for Boolean reachability. 

Our result first appeared in \cite{DBLP:conf/rp/BrihayeBR23}. In this long version, we provide all the proofs.

\paragraph{Related work.}  During the last decade, multiplayer non-zero sum games and their applications to reactive synthesis have raised a growing attention, see for instance the surveys~\cite{DBLP:conf/lata/BrenguierCHPRRS16,DBLP:journals/siglog/Bruyere21,GU08}. When several players (like the followers) play with the aim to satisfy their objectives, several \emph{solution concepts} exist such as NE, subgame perfect equilibrium (SPE)~\cite{selten}, secure equilibria~\cite{DBLP:conf/tacas/ChatterjeeH07,DBLP:journals/tcs/ChatterjeeHJ06}, or admissibility~\cite{Berwanger07,BrenguierRS15}. Several results have been obtained, for Boolean and quantitative objectives, about the constrained existence problem which consists in deciding whether there exists a solution concept such that the payoff obtained by each player is larger than some given threshold. Let us mention~\cite{ConduracheFGR16,Ummels08,UmmelsW11} for results on the constrained existence for NEs and~\cite{DBLP:conf/icalp/BriceRB22,DBLP:conf/csl/BriceRB22,BrihayeBGRB20,Ummels06} for SPEs. Some of them rely on a recent elegant characterization of SPE outcomes~\cite{Raskin2021,DBLP:journals/mor/FleschP17}. 

Stackelberg games with \emph{several followers} have been recently studied in the context of \emph{rational synthesis}: in~\cite{FismanKL10} in a setting where the followers are cooperative with the leader, and later in~\cite{KupfermanPV16} where they are adversarial. Rational responses of the followers are, for instance, to play an NE or an SPE. The rational synthesis problem and the SPS problem are incomparable, as illustrated in~\cite[Section 4.3.2]{phdTamines}: in rational synthesis, each component of the environment acts selfishly, whereas in SPS, the components cooperate in a way to obtain a Pareto-optimal cost. In \cite{DBLP:conf/tacas/KupfermanS22}, the authors solve the rational synthesis problem that consists in deciding whether the leader can announce a strategy satisfying his objective, when the objectives of the players are specified by LTL formulas. Complexity classes for various $\omega$-regular objectives are established in~\cite{ConduracheFGR16} for both cooperative and adversarial settings. Extension to quantitative payoffs, like mean-payoff or discounted sum, is studied in~\cite{GuptaS14,GuptaS14c} in the cooperative setting and in~\cite{BalachanderGR20,FiliotGR20} in the adversarial setting. 

The concept of \emph{rational verification} has been introduced in~\cite{GutierrezNPW20}, where instead of deciding the existence of a strategy for the leader, one verifies that some given leader's strategy satisfies his objective, whatever the NE responses of the followers. An algorithm and its implementation in the EVE system are presented in \cite{GutierrezNPW20} for objectives specified by LTL formulas. This verification problem is studied in~\cite{DBLP:journals/games/GutierrezSW22} for mean-payoff objectives for the followers and an omega-regular objective for the leader, and it is solved in~\cite{DBLP:journals/corr/abs-2301-12913} for both NE and SPE responses of the followers and for a variety of objectives including quantitative objectives. The Stackelberg-Pareto verification problem is solved in~\cite{DBLP:conf/concur/BruyereRT22} for some $\omega$-regular or LTL objectives.

Since our paper~\cite{DBLP:conf/rp/BrihayeBR23}, the rational synthesis and verification problems have been further studied for quantitative reachability objectives in \cite{Christophe}, for both cases of rational responses: NE responses and Pareto-optimal responses.

\paragraph{Structure of the paper.} In Section~\ref{sec:Prelim}, we introduce the concept of Stackelberg-Pareto games with quantitative reachability costs. We also recall several useful related notions. In Section~\ref{sec:bounding}, we show that if there exists a solution to the SPS problem, then there exists one whose Pareto-optimal costs are exponential in the size of the instance. In Section~\ref{sec:complexity}, we prove that the SPS problem is \textsf{NEXPTIME}-complete by using the result of the previous section. Finally, we give a conclusion and some future work. 

\section{Preliminaries and Studied Problem} \label{sec:Prelim}
%=========================================

We introduce the concept of Stackelberg-Pareto games with quantitative reachability costs. We present the related Stackelberg-Pareto synthesis problem and state our main result.

\subsection{Graph Games}
%=======================

%\subsubsection{Game Arenas}
%==========================

\paragraph{Game arenas.} A \emph{game arena} is a tuple $A = (V, V_0, V_1, E, v_0, w)$ where: \emph{(1)} $(V,E)$ is a finite directed graph with $V$ as set of vertices and $E$ as set of edges (it is supposed that every vertex has a successor), \emph{(2)} $V$ is partitioned as $V_0 \cup V_1$ such that $V_0$ (resp. $V_1$) represents the vertices controlled by Player~0 (resp. Player~1), \emph{(3)} $v_0 \in V$ is the initial vertex, and \emph{(4)} $w \colon E \rightarrow \nat$ is a weight function that assigns a non-negative integer\footnote{Notice that null weights are allowed.} to each edge, such that $W = \max_{e \in E}w(e)$ denotes the maximum weight. An arena $A$ is \emph{binary} if $w(e) \in \{0,1\}$ for all $e \in E$.\medskip

%\subsubsection{Plays and Histories}
%==================================

\paragraph{Plays and histories.} 
A \emph{play} in an arena $A$ is an infinite sequence of vertices $\rho = \rho_0\rho_1 \ldots \in V^\omega$ such that $\rho_0 = v_0$ and $(\rho_k,\rho_{k+1}) \in E$ for all $k \in \nat$. \emph{Histories} are finite sequences $h = h_0 \ldots h_k \in V^+$ defined similarly. We denote by $\Last{h}$ the last vertex $h_k$ of the history $h$ and by $|h|$ its length (equal to $k$). Let $\play_A$ denote the set of all plays in $A$, $\hist_A$ the set of all histories in $A$, and $\histi{i}_A$ the set of all histories in $A$ ending on a vertex in $V_i$, $i = 0,1$. The mention of the arena will be omitted when it is clear from the context. If a history $h$ is prefix of a play $\rho$, we denote it by $h \prefix \rho$. Given a play $\rho = \rho_0\rho_1 \ldots$, we denote by $\rho_{\leq k}$ the prefix $\rho_0\ldots \rho_k$ of $\rho$, and by $\rho_{\geq k}$ its suffix $\rho_k \rho_{k+1} \ldots$. We also write $\rho_{[k,\ell]}$ for $\rho_k \ldots \rho_\ell$. The \emph{weight} of $\rho_{[k,\ell]}$ is equal to $w(\rho_{[k,\ell]}) = \Sigma_{j = k}^{\ell-1} w(\rho_j,\rho_{j + 1})$.\medskip

%\subsubsection{Strategies}
%=========================

\paragraph{Strategies.}
Let $i \in \{0,1\}$, a \emph{strategy} for Player~$i$ is a function $\sigma_i \colon \histi{i} \rightarrow V$ assigning to each history $h \in \histi{i}$ a vertex $v = \sigma_i(h)$ such that $(\Last{h},v) \in E$. We denote by $\Sigma_i$ the set of all strategies for Player~$i$. We say that a strategy $\sigma_i$ is \emph{memoryless} if for all $h,h' \in \histi{i}$, if $\Last{h} = \Last{h'}$, then $\sigma_i(h) = \sigma_i(h')$. A strategy is considered \emph{finite-memory} if it can be encoded by a Mealy machine and its \emph{memory size} is the number of states of the machine~\cite{2001automata}.\footnote{We assume that the reader is familiar with the concept of finite-memory strategy and memoryless strategy.}

A play $\rho$ is \emph{consistent} with a strategy $\sigma_i$ if for all $k \in \nat$, $\rho_k \in V_i$ implies that $\rho_{k+1} = \sigma_i(\rho_{\leq k})$. Consistency is extended to histories as expected. We denote $\play_{\sigma_i}$ (resp. $\hist_{\sigma_i}$) the set of all plays (resp. histories) consistent with $\sigma_i$. Given a couple of strategies $(\sigma_0,\sigma_1)$ for Players~0 and~1, there exists a single play that is consistent with both of them, that we denote by $\out(\sigma_0,\sigma_1)$ and call the \emph{outcome} of $(\sigma_0,\sigma_1)$.\medskip

%\subsubsection{Reachability Costs}
%=================================

\paragraph{Reachability costs.} Given an arena $A$, let us consider a non-empty subset $T \subseteq V$ of vertices called \emph{target}. We say that a play $\rho = \rho_0\rho_1 \ldots$ \emph{visits} the target $T$, if $\rho_k \in T$ for some $k$. We define a \emph{cost function} $\textsf{cost}_T \colon \play \rightarrow \overline{\nat}$, where $\overline{\nat} = \nat \cup \{\infty\}$, that assigns to every play $\rho$ the quantity $\costT{T}{\rho} = \min\{w(\rho_{\leq k}) \mid \rho_k \in T\}$, that is, the weight to the first visit of $T$ if $\rho$ visits $T$, and $\infty$ otherwise. The cost function is extended to histories in the expected way. Note that for two histories $g,h$ such that $g \prefix h$ (i.e. $g = h_0 \ldots h_k$, $h = h_0 \ldots h_\ell$, with $k \leq \ell$), it may happen that $\costT{T}{g} = \infty$ and $\costT{T}{h} < \infty$ because $h$ visits $T$ after $g$. However, $w(g) \leq w(h) = w(g) + w(h_{[k,\ell]})$ because the weight function do not take into account the visit or not to $T$. In the proofs of this paper, weight and cost functions should not be confused.

\subsection{Stackelberg-Pareto Synthesis Problem}
%================================================

%\subsubsection{Stackelberg-Pareto Games}
%=======================================

\paragraph{Stackelberg-Pareto games.}
Let $t \in \nat \setminus\{0\}$, a \emph{Stackelberg-Pareto reachability game} (SP game) is a tuple $G = (A,T_0,T_1,\ldots,T_t)$ where $A$ is a game arena and $T_i$ are targets for all $i \in \{0,\ldots,t\}$, such that $T_0$ is Player~0's target and $T_1,\ldots,T_t$ are the $t$ targets of Player~1. When $A$ is binary, we say that $G$ is \emph{binary}. The \emph{dimension} $t$ of $G$ is the number of Player~1's targets, and we denote by $\DimGame{t}$ (resp. $\BinDimGame{t}$) the set of all (resp. binary) SP games with dimension $t$. The notations $\play_G$ and $\hist_G$ may be used instead of $\play_A$ and $\hist_A$.

To distinguish the two players with respect to their targets, we introduce the following terminology. The \emph{cost} of a play $\rho$ is the tuple $\cost{\rho} \in \overline{\nat}^t$ such that $\cost{\rho} = (\costT{T_1}{\rho},\ldots,\costT{T_t}{\rho})$. The \emph{value} of a play $\rho$ is a non-negative integer or $\infty$ defined by $\val{\rho} = \costT{T_0}{\rho}$. The value can be viewed as the score of Player~0 and the cost as the score of Player~1. Both functions are extended to histories in the expected way. In the sequel, given a cost $c \in \overline{\nat}^t$, we denote by $c_i$ the $i$-th component of $c$ and by $c_{min}$ the component of $c$ that is minimum, i.e. $c_{min} = \min \{c_i \mid i \in \{1,\ldots,t\} \}$.

In an SP game, Player~0 wishes to minimize the value of a play with respect to the usual order $<$ on $\nat$ extended to $\overline{\nat}$ such that $n < \infty$ for all $n \in \nat$.  To compare the costs of Player $1$, the following component-wise order is introduced. Let $c, c' \in \overline{\nat}^t$ be two costs, we say that $c \le c'$ if $c_i \le c_i'$ for all $i \in \{1, \ldots, t\}$. Moreover, we write $c < c'$ if $c \le c'$ and $c \ne c'$. Notice that the order defined on costs is not total. Given two plays with respective costs $c$ and $c'$, if $c < c'$, then Player~$1$ prefers the play with lower cost $c$. \medskip

%\subsubsection{Stackelberg-Pareto Synthesis Problem}
%===================================================

\paragraph{Stackelberg-Pareto synthesis problem.}
Given an SP game and a strategy $\sigma_0$ for Player~0, we consider the set $C_{\sigma_0}$ of costs of plays consistent with $\sigma_0$ that are \emph{Pareto-optimal} for Player~$1$, i.e., minimal with respect to the order $\leq$ on costs. Hence, 
$$C_{\sigma_0} = \min\{\cost{\rho} \mid \rho \in \play_{\sigma_0}\}.$$ 
Notice that $C_{\sigma_0}$ is an antichain. A cost $c$ is said to be $\sigma_0$\emph{-fixed Pareto-optimal} if $c \in C_{\sigma_0}$. Similarly, a play is said to be $\sigma_0$-fixed Pareto-optimal if its cost is $\sigma_0$-fixed Pareto-optimal. We will omit the mention of $\sigma_0$ when it is clear from context.

The problem we study is the following one: given an SP game $G$ and a bound $B \in \nat$, is there a strategy $\sigma_0$ for Player~0 such that, for all strategies $\sigma_1$ for Player~1, if the outcome $\out(\sigma_0,\sigma_1)$ is Pareto-optimal, then the value of the outcome is below $B$. This is equivalent to say that for all $\rho \in \play_{\sigma_0}$, if $\cost{\rho}$ is $\sigma_0$-fixed Pareto-optimal, then $\val{\rho}$ is below $B$.

\bigskip \noindent
{\bf Problem. } 
The \emph{Stackelberg-Pareto Synthesis problem} (SPS problem) is to decide, given an SP game $G$ and a bound $B \in \nat$, whether 
\begin{eqnarray} \label{eq:SPS}
\exists \sigma_0 \in \Sigma_0,  \forall \sigma_1 \in \Sigma_1, \cost{\out(\sigma_0,\sigma_1)} \in C_{\sigma_0} \Rightarrow \val{\out(\sigma_0,\sigma_1)} \le B.
\end{eqnarray}

Any strategy $\sigma_0$ satisfying (\ref{eq:SPS}) is called a \emph{solution} and we denote it by $\sigma_0 \in \SPS{G}{B}$. Our main result is the following theorem.  

\begin{theorem} \label{thm:main}
The SPS problem is \textsf{NEXPTIME}-complete.
\end{theorem}

The non-deterministic algorithm is exponential in the number of targets $t$ and in the size of the binary encoding of the maximum weight $W$ and the bound $B$. The general approach to obtain this \textsf{NEXPTIME}-membership is to show that when there is a solution $\sigma_0 \in \SPS{G}{B}$, then there exists one that is finite-memory and whose memory size is exponential. An important part of this paper is devoted to this proof. Then we show that such a strategy can be guessed and checked to be a solution in exponential time.\medskip

%\subsubsection{Example}
%======================

\paragraph{Example.}
To provide a better understanding of the SPS problem, let us solve it on a specific example. The arena $A$ is displayed on Figure~\ref{fig:ex} (left part) where the vertices controlled by Player~0 (resp. Player~$1$) are represented as circles (resp. squares). The weights are indicated only if they are different from 1 (e.g., the edge $(v_0,v_6)$ has a weight of 1). The initial vertex is $v_0$. The target of Player 0 is $T_0 = \{v_3,v_9\}$ and is represented by doubled vertices. Player~\ref{fig:ex} has three targets: $T_1 = \{v_1,v_8\}$, $T_2 = \{v_9\}$ and $T_3 = \{v_2,v_4\}$, that are represented using colors (green for $T_1$, red for $T_2$, blue for $T_3$). Let us exhibit a solution $\sigma_0$ in $\SPS{G}{5}$.

%\begin{figure}[h]
\begin{figure}[t]
\begin{center}
\begin{tikzpicture}[yscale=.6] 
\tikzstyle{carre}=[draw,thick, minimum height=0.5cm, minimum width=0.5cm, inner sep=0pt]
\tikzstyle{rond}=[circle,draw,thick, minimum size=0.5cm, inner sep=0pt]
\tikzstyle{fleche}=[->, >=latex,line width=0.4mm]

% Sommets
\node[carre] (a) at (0,0) {$v_0$};
\node[carre,fill=green] (b) at (2,1) {$v_1$};
\node[rond,fill=cyan] (c) at (4,0) {$v_2$};
\node[rond,double] (e) at (4,2) {$v_3$};
\node[carre] (f) at (6,1) {$v_5$};
\node[rond] (g) at (2,-1) {$v_6$};
\node[carre,fill=cyan] (h) at (6,3) {$v_4$};
\node[rond,fill=green] (i) at (6,-1) {$v_8$};
\node[carre] (j) at (4,-2) {$v_7$};
\node[rond,fill=red,double] (k) at (6,-3) {$v_9$};

% Arcs
\draw[fleche] (-1,0) -- (a);
\draw[fleche] (a) -- (b) node[midway,fill=white]{4};
\draw[fleche] (b) -- (c) node[midway,fill=white]{4};
\draw[fleche] (b) -- (e);
\draw[fleche] (e) -- (f);
\draw[fleche] (e) to[bend right] node[midway,fill=white]{2} (h);
\draw[fleche] (a) -- (g);
\draw[fleche] (h) to[bend right] (e);
\draw[fleche] (g) to[bend right] (j);
\draw[fleche] (j) -- (i);
\draw[fleche] (j) -- (k);

% Self-loop
\path (c) edge [loop right, >=latex, line width=0.4mm] (c);
\path (f) edge [loop right, >=latex, line width=0.4mm] (f);
\path (g) edge [loop right, >=latex, line width=0.4mm] (g);
\path (i) edge [loop right, >=latex, line width=0.4mm] (i);
\path (k) edge [loop right, >=latex, line width=0.4mm] (k);

%\node at (3,-4) {Figure 1: Arena $A$};

\begin{scope}[xshift=11cm,yshift=-1cm]

\node (aw) at (0,6) {$v_0$};
\node (bw) at (1,4.5) {$v_1$};
\node (cw) at (1,3) {$v_3$};
\node (dw) at (1,1.5) {$v_4$};
\node (ccw) at (1,0) {$v_3$};
\node (ddw) at (1,-1.5) {$v_4$};
\node (dddw) at (1,-3.5) {$\vdots$};
\node (ew) at (-1,4.5) {$v_6$};
\node (fw) at (-1,3) {$v_6$};
\node (gw) at (-1,1.5) {$v_7$};
\node (hw) at (-1,0) {$v_9$};
\node (iw) at (-1,-1.5) {$v_9$};
\node (jw) at (-1,-3.5) {$\vdots$};

\draw[->,>=latex, line width=0.4mm] (aw) -- (bw);
\draw[->,>=latex, line width=0.4mm] (bw) -- (cw);
\draw[->,>=latex, line width=0.4mm] (cw) -- (dw);
\draw[->,>=latex, line width=0.4mm] (dw) -- (ccw);
\draw[->,>=latex, line width=0.4mm] (ccw) -- (ddw);
\draw[->,>=latex, line width=0.4mm] (ddw) -- (dddw);
%\draw[->,>=latex, line width=0.4mm] (cw) to[bend right] (dw);
%\draw[->,>=latex, line width=0.4mm] (dw) to[bend right] (cw);
\draw[->,>=latex, line width=0.4mm] (aw) -- (ew);
\draw[->,>=latex, line width=0.4mm] (ew) -- (fw);
\draw[->,>=latex, line width=0.4mm] (fw) -- (gw);
\draw[->,>=latex, line width=0.4mm] (gw) -- (hw);
\draw[->,>=latex, line width=0.4mm] (hw) -- (iw);
\draw[->,>=latex, line width=0.4mm] (iw) -- (jw);

%\path (hw) edge [loop right, >=latex, line width=0.4mm] (hw);
\end{scope}

\end{tikzpicture}
\end{center}
\caption{Arena $A$ (on the left) -- Witness tree (on the right)}
\label{fig:ex}
\label{example}
\end{figure}

We define $\sigma_0$ as the strategy that always moves from $v_3$ to $v_4$, and that loops once on $v_6$ and then moves to $v_7$. The plays consistent with $\sigma_0$ are $v_0v_1v_2^\omega$, $v_0v_1(v_3v_4)^\omega$, $v_0v_6v_6v_7v_8^\omega$, and $v_0v_6v_6v_7v_9^\omega$. The Pareto-optimal plays are $v_0v_1(v_3v_4)^\omega$ and $v_0v_6v_6v_7v_9^\omega$ with respective costs $(4,\infty,7)$ and $(\infty,4,\infty)$, and they both yield a value less than or equal to 5. Notice that $\sigma_0$ has to loop once on $v_6$, i.e., it is not memoryless\footnote{One can prove that there exists no memoryless solution.}, otherwise the consistent play $v_0v_6v_7v_8^\omega$ has a Pareto-optimal cost of $(3,\infty,\infty)$ and an infinite value. 

Interestingly, the Boolean version of this game does not admit any solution. In this case, given a target, the player's goal is simply to visit it (and not to minimize the cost to reach it). That is, the Boolean version is equivalent to the quantitative one with all weights and the bound $B$ put to zero. In the example, the play $v_0v_1v_2^\omega$ is Pareto-optimal (with visits to $T_1$ and $T_3$), whatever the strategy of Player~$0$, and this play does not visit Player~$0$'s target.\medskip

\subsection{Tools}
%==================

We here present two tools used in the proofs for solving the SPS problem.

\paragraph{Witnesses.}
An important tool is the concept of witness~\cite{DBLP:conf/concur/BruyereRT21}. Given a solution $\sigma_0$, for all $c \in C_{\sigma_0}$, we can choose arbitrarily a play $\rho$ called \emph{witness} of the cost $c$ such that $\cost{\rho} = c$. The set of all chosen\footnote{Note that the witness set is not necessarily unique.} witnesses is denoted by $\wit_{\sigma_0}$, whose size is the size of $C_{\sigma_0}$. Since $\sigma_0$ is a solution, the value of each witness is below $B$.
We define the \emph{length} of a witness $\rho$ as the length 
$$\length{\rho} = \min \{ |h| \mid h \prefix \rho \wedge \cost{h} = \cost{\rho} \wedge \val{h} = \val{\rho} \}.$$ 
Hence it is the length of the shortest history $h$ that visits the same targets as $\rho$. The \emph{length} of $\wit_{\sigma_0}$, denoted by $\length{\wit_{\sigma_0}}$, is equal to $\Sigma_{\rho \in \wit_{\sigma_0}} \length{\rho}$. Moreover, given $h \in \hist_{\sigma_0}$, we write $\wit_{\sigma_0}(h)$ the set of witnesses for which $h$ is a prefix, i.e., 
$$\wit_{\sigma_0}(h) = \{\rho \in \wit_{\sigma_0} \mid h \prefix \rho\}.$$ 
Notice that $\wit_{\sigma_0}(h) = \wit_{\sigma_0}$ when $h = v_0$, and that the size of $\wit_{\sigma_0}(h)$ decreases as the size of $h$ increases, until it contains a single play or becomes empty.

It is useful to see the set $\wit_{\sigma_0}$ as a tree composed of $|\wit_{\sigma_0}|$ infinite branches (corresponding to the witnesses). The following notions about this tree  will be useful. We say that a history $h$ is a \emph{branching point} if there are two witnesses whose greatest common prefix is $h$, that is, there exists $v \in V$ such that $0 < |\wit_{\sigma_0}(hv)| < |\wit_{\sigma_0}(h)|$. Given a witness $\rho$, we define the following equivalence relations $\sim$ on histories that are prefixes of $\rho$: 
$$h \sim h' \quad\Leftrightarrow\quad (\val{h},\cost{h},\wit_{\sigma_0}(h)) = (\val{h'},\cost{h'},\wit_{\sigma_0}(h')).$$ Notice that if $h \sim h'$, then either $h \prefix h'$ or $h' \prefix h$ and no new target is visited and no branching point is crossed from the shortest history to the longest one. We call \emph{region} of $h$ its equivalence class. This leads to a \emph{region decomposition} of each witness $\rho$, such that the first region is the region of the initial vertex $v_0$ and the last region is the region of $h \prefix \rho$ such that $|h| = \length{\rho}$. Finally, a \emph{deviation} is a history $hv$ with $h \in \histi{1}$ and $v\in V$, such that $h$ is prefix of some witness, but $hv$ is prefix of no witness. 

We illustrate these different notions on the previous example and its solution $\sigma_0$. A set of witnesses is $\wit_{\sigma_0} = \{v_0v_1(v_3v_4)^\omega, v_0v_6v_6v_7v_9^\omega\}$ depicted on Figure~\ref{fig:ex} (right part). We have that $\length{v_0v_6v_6v_7v_9^\omega} = |v_0v_6v_6v_7v_9|=4$ and $\length{\wit_{\sigma_0}} = 7$. Moreover, $\wit_{\sigma_0}(v_0) = \wit_{\sigma_0}$, $\wit_{\sigma_0}(v_0v_1) = \{v_0v_1(v_3v_4)^\omega\}$, and $\wit_{\sigma_0}(v_0v_1v_2) = \emptyset$. The initial vertex $v_0$ is a branching point, $v_0v_1v_2$ is a deviation, and the region decomposition of the witness $v_0v_6v_6v_7v_9^\omega$ is equal to $\{v_0\}$, $\{v_0v_6,v_0v_6v_6,v_0v_6v_6v_7\}$, $\{v_0v_6v_6v_7v_9^k \mid k \geq 1\}$.  \medskip

%\subsubsection{Reduction to Binary Arenas}
%==========================================

\paragraph{Reduction to binary arenas.}
Working with general arenas requires to deal with the parameter $W$ in most of the proofs. To simplify the arguments, we reduce the SPS problem to \emph{binary} arenas, by replacing each edge with a weight $w \geq 2$ by a path of $w$ edges of weight $1$. This (standard) reduction is exponential, but only in the size of the binary encoding of $W$.

\begin{lemma} \label{lem:binary}
Let $G = (A, T_0, \ldots , T_t)$ be an SP game and $B \in \nat$. Then one can construct in exponential time an SP game $G' = (A', T_0, \ldots , T_t)$ with a binary arena $A'$ such that
\begin{itemize}
    \item the set of vertices $V'$ of $A'$ contains $V$ and has size $|V'| \le |V| \cdot W$,
    \item there exists a one-to-one correspondance $g$ between the set of strategies on $A$ and the set of strategies on $A'$. Moreover $\sigma_0$ is a solution in $\SPS{G}{B}$ if and only if $g(\sigma_0)$ is a solution in $\SPS{G'}{B}$.   %there exists a solution in $\SPS{G}{B}$ if and only if there exists a solution in $\SPS{G'}{B}$.
    \end{itemize}
\end{lemma}

\begin{proof}
Let $X = \{e \in E \mid w(e) \geq 2 \}$ be the set of edges of $A$ with weight at least 2. For each $e \in X$, we replace $e$ by a succession of $w(e)-1$ new vertices linked by $w(e)$ new edges of weight 1, and for each $e \in E \setminus X$, we keep $e$ unmodified. In this way, we get a directed graph $(V',E')$ with a weight function $w' \colon E' \rightarrow \{0,1\}$ and with a size $|V'| \le |V| \cdot W$. Notice that if $X = \emptyset$, the arena is already binary. The new arena $A' = (V',V'_0,V'_1,E',v_0,w')$ has the same initial vertex $v_0$ as $A$, and a partition $V'_0 \cup V'_1$ of $V'$ such that each new vertex is added\footnote{Each new vertex could be added to $V_1$ as it has a unique successor.} to $V_0$, hence $V_0 \subseteq V'_0$, $V_1 = V'_1$. The new SP game $G' = (A', T_0, \ldots , T_t)$ keeps the same targets as in $G$.  
%Let $e := (v^e_0,v^e_{w(e)}) \in X$.  We set $V' = V \cup \{v^e_i\}_{1\le i \le w(e)-1}$ and $E' = E \setminus \{e\} \cup \{(v^e_i,v^e_{i+1})\}_{0\le i \le w(e)-1}$. We define a new weight function $w': E' \rightarrow \nat$ such that, for all $e' \in E'$, $w(e') = w(e')$ if $e' \in E$ and $w(e') = 1$ otherwise. Moreover, we build a new partition $(V'_0 V'_1)$ such that $V_0 \subseteq V_0$, $V_1 \subseteq V'_1$ and $\forall i \in \{1,2,\ldots,w(e)-2,w(e)-1\}, v^e_i \in V_0' \iff v^e_0 \in V_0$. Now consider the arena $A' = (V',V'_0,V'_1,E',v_0,w')$.

 Clearly, there is a trivial bijection $f$ from $\play_A$ to $\play_{A'}$. Indeed, it suffices to replace all edges $e$ of a play with weight $w(e) \geq 2$ by the corresponding new path composed of $w(e)$ edges of weight 1. This bijection preserves the cost and the value of the plays. Moreover, there also exists a bijection $g$ from the set of strategies on $A$ to the set of strategies on $A'$ as each new vertex has a unique successor. Notice that $g$ is coherent with $f$, i.e.,  $\out(g(\sigma_0),g(\sigma_1)) = f(\out(\sigma_0,\sigma_1))$ for all $\sigma_0 \in \Sigma_0$ and all $\sigma_1 \in \Sigma_1$. Therefore, for all $\sigma_0 \in \Sigma_0$, we get that $\sigma_0 \in \SPS{G}{B}$ if and only if $g(\sigma_0) \in \SPS{G'}{B}$. 
%
%We can repeat the same construction for all $e \in X$ and, by transitivity, $\SPS{G}{B} = \emptyset \iff \SPS{(A',T_0,\ldots,T_t)}{B} = \emptyset$. Moreover, for each $e \in X$, we add $w(e) - 1$ vertices to $V$, thus, $|V'| \le W|V|$.
\end{proof}

The transformation of the arena $A$ into a binary arena $A'$ has consequences on the size of the SPS problem instance. Since the weights are encoded in binary, the size $|V'|$ could be exponential in the size $|V|$ of the original instance. However, this will have \emph{no impact} on our main result because $|V|$ never appears in the exponent in our calculations (this will be detailed in the proof of Theorem~\ref{thm:main}). 

Note that along any history $h = h_0 \ldots h_\ell$ in a binary arena, the total weight increases by a weight of 0 or 1 from $w(h_0)$ to $w(h)$, that is $\{w(h_{[0, k]}) \mid 0 \leq k \leq \ell \} = \{0,1, \ldots, w(h)\}$. This property will be used in some proofs.

\section{Improving a Solution}
%==============================

In order to solve efficiently the SPS problem, we would like to work with solutions resulting in small costs. To do so, in this section, we introduce some techniques that allow to modify a solution to the SPS problem into a solution with smaller Pareto-optimal costs.

\subsection{Order on Strategies and Subgames}
%================================================

Given two strategies $\sigma_0, \sigma_0'$ for Player~0, we define $\sigma_0' \preceq \sigma_0$ if for all $c \in C_{\sigma_0}$, there exists $c' \in C_{\sigma_0'}$ such that $c' \le c$, and we say that $\sigma'_0$ is \emph{better} than $\sigma_0$. This relation $\preceq$ on strategies is a preorder (it is reflexive and transitive). %\commG{I believe that $\le$ is anti-symmetric (thus an order) on the sets of Pareto-optimal costs because they are antichains. Let $C_0$ and $C_1$ two antichains such that $C_0 \le C_1$ and $C_1 \le C_0$. Let $c_0 \in C_0$, as $C_1 \le C_0$, there exists $c_1 \in C_1$ such that $c_1 \le c_0$. Since $C_0 \le C_1$, there exists $c_0' \in C_0$ such that $c_0' \le c_1$ thus $c_0' \le c_0$, but $C_0$ is an antichain so $c_0 = c_0' = c_1 \in C_1$. As a result, $C_0 \subseteq C_1$ and so, by symmetry, $C_0 = C_1$.}. 
Given two strategies $\sigma_0, \sigma'_0$ such that $\sigma'_0 \preceq \sigma_0$, we define $\sigma'_0 \prec \sigma_0$ when $C_{\sigma'_0} \ne C_{\sigma_0}$, and 
$\sigma'_0 \simeq \sigma_0$ when $C_{\sigma'_0} = C_{\sigma_0}$. In the sequel, we  modify solutions $\sigma_0$ to the SPS problem to get better solutions $\sigma'_0 \preceq \sigma_0$, and we say that $\sigma'_0$ \emph{improves} the given solution $\sigma_0$. 

A \emph{subgame} of an SP game $G$ is a couple $(G,h)$, denoted $G_{|h}$, where $h \in \hist$. In the same way that $G$ can be seen as the set of its plays, $G_{|h}$ is seen as the restriction of $G$ to plays with the prefix $h$ (see Figure~\ref{fig:subgame}). In particular, we have $G_{|v_0} = G$ where $v_0$ is the initial vertex of $G$. The value and cost of a play $\rho$ in $G_{|h}$ are the same as those of $\rho$ as a play in $G$. The \emph{dimension} of $G_{|h}$ is the dimension of $G$ minus the number of targets visited\footnote{Notice that we do not include $\Last{h}$ in $h'$, as it can be seen as the initial vertex of $G_{|h}$.} by $h'$ such that $h'\Last{h} = h$.

\begin{figure}
    \centering
    \begin{tikzpicture}[x=0.75pt,y=0.75pt,yscale=-1,xscale=1,scale=1, every node/.style={scale=1}]
    \draw [color={rgb, 255:red, 0; green, 0; blue, 0 }  ,draw opacity=1 ][line width=1.5]    (346.39,27.11) -- (464.29,233.03) ;
    \draw [line width=1.5]    (322.27,27.11) -- (202.78,234.23) ;
    \draw  [line width=1.5]  (346.39,27.24) .. controls (346.39,20.6) and (341,15.21) .. (334.36,15.21) -- (334.3,15.21) .. controls (327.66,15.21) and (322.27,20.6) .. (322.27,27.24) -- (322.27,27.24) .. controls (322.27,33.89) and (327.66,39.28) .. (334.3,39.28) -- (334.36,39.28) .. controls (341,39.28) and (346.39,33.89) .. (346.39,27.24) -- cycle ;
    \draw  [line width=0.75] [line join = round][line cap = round] (337.8,136.95) .. controls (337.07,145.24) and (333.86,154) .. (335.6,161.82) .. controls (338.22,173.62) and (347.49,178.52) .. (349.71,190.58) .. controls (351.42,199.91) and (347.36,209.76) .. (346.18,219.34) ;
    \draw    (345,130.1) -- (397.38,228.05) ;
    \draw    (331.13,130.1) -- (280.95,227.65) ;
    \draw   (331.13,130.1) .. controls (331.13,126.32) and (334.24,123.25) .. (338.07,123.25) .. controls (341.9,123.25) and (345,126.32) .. (345,130.1) .. controls (345,133.88) and (341.9,136.94) .. (338.07,136.94) .. controls (334.24,136.94) and (331.13,133.88) .. (331.13,130.1) -- cycle ;
    \draw  [line width=0.75] [line join = round][line cap = round] (334.87,39.93) .. controls (311.76,65.99) and (352.11,95.23) .. (338.12,122.83) ;
    \draw  [line width=0.75] [line join = round][line cap = round] (332.03,76.94) .. controls (325.69,103.46) and (297.61,118.7) .. (281.79,138.63) .. controls (263.73,161.37) and (281.79,200.08) .. (255.95,219.16) ;

    \draw (325.73,22.7) node [anchor=north west][inner sep=0.75pt]  [font=\normalsize] [align=left] {$v_{0}$};
    \draw (345.72,77.2) node [anchor=north west][inner sep=0.75pt]  [font=\normalsize] [align=left] {$h$};
    \draw (336.22,220.33) node [anchor=north west][inner sep=0.75pt]  [font=\normalsize] [align=left] {$\rho$};
    \draw (244.28,221.42) node [anchor=north west][inner sep=0.75pt]  [font=\normalsize] [align=left] {$\rho'$};
    \draw (455.53,179.59) node [anchor=north west][inner sep=0.75pt]  [font=\normalsize] [align=left] {$G$};
    \draw (402.73,219.35) node [anchor=north west][inner sep=0.75pt]  [font=\normalsize] [align=left] {$G_{|h}$};
    \end{tikzpicture}
    
    \caption{A game $G$ seen as the set of its plays (containing $\rho$ and $\rho'$), and its subgame $G_{|h}$}
    \label{fig:subgame}
\end{figure}

A strategy for Player~$0$ on $G_{|h}$ is a strategy $\tau_0$ that is only defined for the histories $h' \in \hist$ such that $h \prefix h'$. We denote $\Sigma_{0|h}$ the set of those strategies. Given a strategy $\sigma_0$ for Player~0 in $G$ and $h \in \hist_{\sigma_0}$, we denote the restriction of $\sigma_0$ to $G_{|h}$ by the strategy $\sigma_0{}_{|h}$. Moreover, given $\tau_0 \in \Sigma_{0|h}$, we can define a new strategy $\sigma_0[h\rightarrow\tau_0]$ from $\sigma_0$ as the strategy on $G$ which consists in playing the strategy $\sigma_0$ everywhere, except in the subgame $G_{|h}$ where $\tau_0$ is played. That is, $\sigma_0[h\rightarrow\tau_0](h') = \sigma_0(h')$ if $h \not\prefix h'$, and $\sigma_0[h\rightarrow\tau_0](h') = \tau_0(h')$ otherwise.

As done with $\SPS{G}{B}$, we denote by $\SPS{G_{|h}}{B}$ the set of all solutions $\tau_0 \in \Sigma_{0|h}$ to the SPS problem for the subgame $G_{|h}$ and the bound $B$.

\subsection{Improving a Solution in a Subgame}
%==============================================

A natural way to improve a strategy is to improve it on a subgame. Moreover, if it is a solution to the SPS problem, it is also the case for the improved strategy. 

\begin{lemma} \label{lem:improve}
Let $G$ be an SP game, $B \in \nat$, and $\sigma_0 \in  \SPS{G}{B}$ be a solution. Consider a history $h \in \hist_{\sigma_0}$ and a strategy $\tau_0 \in \Sigma_{0|h}$ in the subgame $G_{|h}$ such that $\tau_0 \prec \sigma_{0|h}$ and $\tau_0 \in \SPS{G_{|h}}{B}$. Then the strategy $\sigma_0' = \sigma_0[h \rightarrow \tau_0]$ is a solution in $\SPS{G}{B}$ and $\sigma_0' \prec \sigma_0$.
\end{lemma}

\begin{proof}
Let us first prove that $\sigma_0' \preceq \sigma_0$, that is, for all $c \in C_{\sigma_0}$, there exists $c' \in C_{\sigma_0'}$ such that $c' \le c$. Let $c \in C_{\sigma_0}$ and $\rho \in \play_{G,\sigma_0}$ be such that $\cost{\rho} = c$. 
\begin{itemize}
\item If $h \not\prefix \rho$, then $\rho$ is also consistent with $\sigma_0'$, thus there exists $c' \in C_{\sigma_0'}$ such that $c' \le c = \cost{\rho}$ by definition of $C_{\sigma_0'}$. \item Otherwise, $h \prefix \rho$. Hence $\rho$ is a play in the subgame ${G_{|h}}$ with a Pareto-optimal cost $c \in C_{\sigma_{0|h}}$. By hypothesis, we have $\tau_0 \prec \sigma_{0|h}$, therefore there exists $c' \in C_{\tau_0}$ such that $c' \le c$ (for some $c$, the inequality is strict: $c' < c$). As $c' \in C_{\tau_0}$, there exists $\rho'  \in \play_{G_{|h}, \tau_0}$ such that $h \prefix \rho'$ and $\cost{\rho'}=c'$. By definition of $\sigma'_0$, we also have that $\rho' \in \play_{G, \sigma'_0}$. Hence, by definition of $C_{\sigma_0'}$, there exists $c'' \in C_{\sigma'_0}$ such that $c'' \leq c'$, and thus $c'' \leq c$. We have thus proved that $\sigma_0' \preceq \sigma_0$.
\end{itemize}
Notice that we have $\sigma_0' \prec \sigma_0$, by the arguments of the second item, as $\tau_0 \prec \sigma_{0|h}$. 

Let us now show that $\sigma'_0 \in  \SPS{G}{B}$. Let $\rho' \in \play_{G,\sigma_0'}$ be such that $c' = \cost{\rho'} \in C_{\sigma_0'}$. We have to prove that $\val{\rho'} \leq B$. If $h \not\prefix \rho'$, then $\rho' \in \play_{G,\sigma_0}$, thus there exists $c \in C_{\sigma_0}$ such that $c \le c'$ by definition of $C_{\sigma_0}$. Since $\sigma_0' \preceq \sigma_0$ by the first part of the proof, it follows that $c = c' = \cost{\rho'} \in C_{\sigma_0}$. Now, recall that $\sigma_0$ is a solution in SPS($G$,$B$), implying that $\val{\rho'} \le B$. If $h \prefix \rho'$, then $\rho' \in \play_{G_{|h},\tau_0}$. As $c' \in C_{\sigma_0'}$, we have $c' \in C_{\tau_0}$ (it is not possible to have $c'' \in C_{\tau_0}$ such that $c'' < c'$ by definition of $\sigma'_0$). As $\tau_0 \in \SPS{G_{|h}}{B}$, it follows that $\val{\rho'} \le B$. In any case, $\val{\rho'} \le B$ showing that $\sigma_0'$ is a solution to the SPS problem in SPS($G$,$B$).
\end{proof}

\subsection{Eliminating a Cycle in a Witness}
%===============================

Another way to improve solutions to the SPS problem is to delete some particular cycles occurring in witnesses as explained in the next lemma.

\begin{lemma} \label{lem:cycle}
Let $G$ be an SP game, $B \in \nat$, and $\sigma_0 \in \SPS{G}{B}$ be a solution. Suppose that in a witness $\rho = \rho_0\rho_1 \ldots \in \wit_{\sigma_0}$, there exist $m, n \in \nat$ such that 
\begin{itemize}
    \item $m < n < \length{\rho}$ and $\rho_m = \rho_n$,
    \item $\rho_{\leq m}$ and $\rho_{\leq n}$ belong to the same region, and
    \item if $\val{\rho_{\leq m}} = \infty$, then the weight $w(\rho_{[m,n]})$ is null.
\end{itemize}
Let $\sigma_0'$ be the strategy defined such that $\sigma'_0(h) = \sigma_0(h)$ if $\rho_{\leq m} \not\prefix h$, and $\sigma'_0(h) = \sigma_0(\rho_{\leq n}h')$ if $h = \rho_{\leq m}h'$. Then $\sigma'_0$ is a solution in $\SPS{G}{B}$ such that $\sigma_0' \preceq \sigma_0$. Moreover, if $\sigma_0' \simeq \sigma_0$, then $\sigma'_0$ has a witness set $\wit_{\sigma'_0}$ such that $\length{\wit_{\sigma'_0}} < \length{\wit_{\sigma_0}}$.
\end{lemma}

Given a witness $\rho$, the first condition of the lemma means that $\rho_{[m,n]}$ is a cycle and that it appears before the last visit of a target by $\rho$. The second one says that $\rho_{\leq m} \sim \rho_{\leq n}$, i.e., no new target is visited and no branching point is crossed from history $\rho_{\leq m}$ to history $\rho_{\leq n}$. The third one says that if $\rho_{\leq m}$ does not visit Player~$0$'s target, then the cycle $\rho_{[m,n]}$ must have a null weight. The new strategy $\sigma'_0$ is obtained from $\sigma_0$ by playing after $\rho_{\leq m}$ as playing after $\rho_{\leq n}$ (thus deleting the cycle $\rho_{[m,n]}$). See Figure~\ref{fig:cycle}.

\begin{figure}
    \centering
        \begin{tikzpicture}[x=0.75pt,y=0.75pt,yscale=-1,xscale=1,scale=1,every node/.style={scale=1}]
    \draw [color={rgb, 255:red, 0; green, 0; blue, 0 }  ,draw opacity=1 ][line width=1.5]    (346.39,27.24) -- (452.37,213.27) ;
    \draw [line width=1.5]    (322.27,27.11) -- (217.47,213.03) ;
    \draw  [line width=1.5]  (346.39,27.24) .. controls (346.39,20.6) and (341,15.21) .. (334.36,15.21) -- (334.3,15.21) .. controls (327.66,15.21) and (322.27,20.6) .. (322.27,27.24) -- (322.27,27.24) .. controls (322.27,33.89) and (327.66,39.28) .. (334.3,39.28) -- (334.36,39.28) .. controls (341,39.28) and (346.39,33.89) .. (346.39,27.24) -- cycle ;
    \draw  [line width=0.75] [line join = round][line cap = round] (339.75,152.86) .. controls (339.02,159.22) and (335.81,165.94) .. (337.55,171.93) .. controls (340.17,180.97) and (349.44,184.73) .. (351.66,193.97) .. controls (353.37,201.12) and (349.31,208.67) .. (348.13,216.01) ;
    \draw    (347.8,89.3) -- (359.15,110.52) -- (401.47,187.67) ;
    \draw    (324.33,90.5) -- (274.15,188.05) ;
    \draw  [line width=0.75]  (347.92,85.68) .. controls (347.92,79.03) and (342.54,73.64) .. (335.89,73.64) -- (335.84,73.64) .. controls (329.19,73.64) and (323.8,79.03) .. (323.8,85.68) -- (323.8,85.68) .. controls (323.8,92.32) and (329.19,97.71) .. (335.84,97.71) -- (335.89,97.71) .. controls (342.54,97.71) and (347.92,92.32) .. (347.92,85.68) -- cycle ;
    \draw  [line width=0.75] [line join = round][line cap = round] (333.3,38.63) .. controls (328.89,51.87) and (325.43,61.79) .. (332.9,74.23) ;
    \draw  [line width=0.75] [line join = round][line cap = round] (337.3,97.43) .. controls (344.22,105.34) and (347.97,119.28) .. (342.9,129.43) ;
    \draw    (351.29,144.5) -- (383.53,204.88) ;
    \draw    (327.82,146.1) -- (297.28,207.38) ;
    \draw  [line width=0.75]  (351.41,140.88) .. controls (351.41,134.23) and (346.02,128.84) .. (339.38,128.84) -- (339.32,128.84) .. controls (332.68,128.84) and (327.29,134.23) .. (327.29,140.88) -- (327.29,140.88) .. controls (327.29,147.52) and (332.68,152.91) .. (339.32,152.91) -- (339.38,152.91) .. controls (346.02,152.91) and (351.41,147.52) .. (351.41,140.88) -- cycle ;
    \draw [color={rgb, 255:red, 0; green, 0; blue, 0 }  ,draw opacity=0 ][fill={rgb, 255:red, 155; green, 155; blue, 155 }  ,fill opacity=0.32 ][line width=0.75] [line join = round][line cap = round]   (323.82,92.76) .. controls (312.78,112.08) and (300.53,136.83) .. (292.81,152.86) .. controls (288.04,162.2) and (281.28,174.63) .. (275.78,185.88) .. controls (269.78,199.13) and (282.21,202.63) .. (285.31,205.11) .. controls (289.61,207.44) and (292.81,207.69) .. (295.25,206.74) .. controls (301.19,204.43) and (302.69,195.01) .. (304.81,191.11) .. controls (312.81,173.86) and (319.81,161.36) .. (327.82,146.1) .. controls (328.57,143.17) and (327.28,141.13) .. (328.03,135.88) .. controls (332.99,130.41) and (337.28,126.78) .. (344.06,129.61) .. controls (347.8,132.7) and (352.5,139.32) .. (351.29,144.5) .. controls (354.53,150.38) and (355.03,151.38) .. (357.31,156.11) .. controls (364.78,169.38) and (381.16,200.18) .. (383.03,202.63) .. controls (384.75,204.88) and (389.53,201.38) .. (397.03,195.88) .. controls (407.28,186.63) and (393.53,175.38) .. (389.49,167.1) .. controls (381.78,152.83) and (373.77,137.77) .. (365.78,123.33) .. controls (362.99,118.29) and (351.53,96.83) .. (347.8,89.3) .. controls (346.28,93.83) and (346.03,93.58) .. (341.53,96.58) .. controls (335.53,98.08) and (329.78,97.58) .. (324.33,90.5) ;

    \draw (325.73,22.7) node [anchor=north west][inner sep=0.75pt]  [font=\normalsize] [align=left] {$v_{0}$};
    \draw (339.82,220.33) node [anchor=north west][inner sep=0.75pt]  [font=\normalsize] [align=left] {$\rho$};
    \draw (451.53,178.39) node [anchor=north west][inner sep=0.75pt]  [font=\normalsize] [align=left] {$G$};
    \draw (325.52,80.88) node [anchor=north west][inner sep=0.75pt]  [font=\normalsize] [align=left] {$\rho_{m}$};
    \draw (330.5,135.58) node [anchor=north west][inner sep=0.75pt]  [font=\normalsize] [align=left] {$\rho_{n}$};
    \end{tikzpicture}
    \caption{Illustration of Lemma~\ref{lem:cycle}}
    \label{fig:cycle}
\end{figure}

From now on, we say that we can \emph{eliminate cycles} according to this lemma\footnote{These are the cycles satisfying the lemma, and not just any cycle.} without explicitly building the new strategy. We also say that a solution $\sigma_0$ is \emph{without cycles} if it does not satisfy the hypotheses of Lemma~\ref{lem:cycle}, i.e., if it is impossible to eliminate cycles to get a better solution.

\begin{proof}[Proof of Lemma~\ref{lem:cycle}] 
Let $g = \rho_{\leq m}$ and $h = \rho_{\leq n}$. We introduce the following notation: for each play $\pi = \pi_0 \pi_1 \ldots \in \play_G$ such that $h \prefix \pi$, we denote by $\bar \pi$ the play $g \pi_{\geq n+1}$, that is, we delete the cycle $\rho_{[m,n]}$ in $\pi$ (we only keep the vertex $\rho_m$). 

We first prove that $\sigma'_0 \preceq \sigma_0$. Let $c \in C_{\sigma_0}$ and $\pi \in \wit_{\sigma_0}$ be a witness with $\cost{\pi} = c$. Let us prove that there exists $c' \in C_{\sigma'_0}$ such that $c' \leq c$. If $g \not\prefix \pi$, then $\pi \in \play_{\sigma'_0}$ by definition of $\sigma'_0$, and thus there exists $c' \in C_{\sigma'_0}$ such that $c' \leq c$ by definition of $C_{\sigma'_0}$. If $g \prefix \pi$, as $g \sim h$ by hypothesis, we have that $\wit_{\sigma_0}(g) = \wit_{\sigma_0}(h)$ and no new target is visited from $g$ to $h$.  
It follows that $h \prefix \pi$ and $\bar \pi \in \play_{\sigma'_0}$ with $\cost{\bar \pi} = \cost{\pi} - w(\rho_{[m,n]}) \leq \cost{\pi}$. By definition of $C_{\sigma'_0}$, there exists $c' \in C_{\sigma'_0}$ such that $c' \leq \cost{\bar \pi}$, and therefore $c' \leq \cost{\pi} = c$. Hence $\sigma'_0 \preceq \sigma_0$. 

We then prove that, if $\sigma'_0 \simeq \sigma_0$ (that is $C_{\sigma'_0} = C_{\sigma_0}$), then $\sigma'_0$ has a witness set $\wit_{\sigma'_0}$ such that $\length{\wit_{\sigma'_0}} < \length{\wit_{\sigma_0}}$. Note that we necessarily have $w(\rho_{[m,n]} = 0)$, as $C_{\sigma'_0} = C_{\sigma_0}$. Let us show that the set 
$$\textsf{W} = \{\pi \mid \pi \in  \wit_{\sigma_0} \mbox{ and } g \not\prefix \pi \} \cup \{\bar \pi \mid \pi \in  \wit_{\sigma_0} \mbox{ and } g \prefix \pi \}$$ 
is the required witness set $\wit_{\sigma'_0}$.
\begin{itemize}
\item First note that in this set $\textsf{W}$, $g \prefix \pi$ implies $h \prefix \pi$ (as already explained in the first part of the proof). Hence, the number of plays in $\textsf{W}$ is equal to $|\wit_{\sigma_0}|$ and their costs constitute the $C_{\sigma_0} = C_{\sigma'_0}$. Moreover, by definition of $\sigma'_0$, $\textsf{W}$ is composed of plays consistent with $\sigma'_0$. 
\item Second, let us prove that $\textsf{W}$ is a witness set for $\sigma'_0$. We have to prove that for any play $\rho' \in \play_{\sigma'_0}$, there exists $\pi' \in \textsf{W}$ such that $\cost{\pi'} \leq \cost{\rho'}$. If $g \not\prefix \rho'$, then $\rho' \in \play_{\sigma_0}$ and there exists $\pi \in \wit_{\sigma_0}$ such that $\cost{\pi} \leq \cost{\rho'}$. By definition of $\textsf{W}$ and as 
$w(\rho_{[m,n]} = 0)$, there exists $\pi' \in \textsf{W}$ such that $\cost{\pi'} = \cost{\pi} \leq \cost{\rho'}$. If $g \prefix \rho'$, we consider $\rho'' \in \play_{\sigma_0}$ such that $\bar \rho'' = \rho'$. Again there exists $\pi \in \wit_{\sigma_0}$ such that $\cost{\pi} \leq \cost{\rho''} = \cost{\rho'}$, and there exists $\pi' \in \textsf{W}$ such that $\cost{\pi'} = \cost{\pi} \leq \cost{\rho'}$.
\item Third, by definition of $\textsf{W}$, we clearly have $\length{\textsf{W}} < \length{\wit_{\sigma_0}}$. 
\end{itemize} 

Finally, we prove that $\sigma'_0$ is a solution to the SPS problem. We have to show that each play $\pi' \in \play_{\sigma'_0}$ with a cost $c' \in C_{\sigma'_0}$ has a value $\val{\pi'} \leq B$. If $g \not\prefix \pi'$, then $\pi' \in \play_{\sigma_0}$. Notice that $c' = \cost{\pi'} \in C_{\sigma_0}$ because $\sigma'_0 \preceq \sigma_0$. Therefore, as $\sigma_0 \in \SPS{G}{B}$, we get that $\val{\pi'} \leq B$. If $g \prefix \pi'$, then either $\val{g} < \infty$, or $\val{g} = \infty$ in which case the weight of $\rho_{[m,n]}$ is null by hypothesis. 
\begin{itemize}
\item In the first case, $\val{g} \leq B$ because $g$ is prefix of the witness $\rho \in \wit_{\sigma_0}$ and $\sigma_0 \in \SPS{G}{B}$. This shows that $\val{\pi'} \leq B$. 
\item In the second case, we consider $\pi \in \play_{\sigma_0}$ such that $\bar \pi = \pi'$. Notice that $\cost{\pi} = c'$ because $\rho_{[m,n]}$ has a null weight and $g \sim h$ (in particular, $\cost{g} = \cost{h}$, that is, no new Player 1's target is visited from $g$ to $h$). We get that $c' = \cost{\pi} \in C_{\sigma_0}$ since $\sigma'_0 \preceq \sigma_0$. It follows that $\val{\pi} \leq B$ as $\sigma_0 \in \SPS{G}{B}$. As $g \sim h$ (in particular, $\val{g} = \val{h}$), as $\val{g} = \infty$, $\pi$ visits Player~$0$'s target after $h$, thus not along $\rho_{[m,n]}$. It follows that $\val{\pi'} \leq B$. Therefore $\sigma'_0 \in \SPS{G}{B}$.
\end{itemize}
This shows that $\sigma'_0$ is a solution to the SPS problem.
\end{proof}

\section{Bounding Pareto-Optimal Payoffs} 
\label{sec:bounding}
%=========================================

In this section, we show that if there exists a solution to the SPS problem, then there exists one whose Pareto-optimal costs are exponential in the size of the instance (see Theorem~\ref{thm:bounded-costs} below). It is a \emph{crucial step} to prove that the SPS problem is in \textsf{NEXPTIME}. This is the main contribution of this paper. Its proof is partly based on the lemmas of the previous section.

\begin{theorem} \label{thm:bounded-costs} 
Let $G \in \BinDimGame{t}$ be a binary SP game with dimension $t$, $B \in \nat$, and $\sigma_0 \in \mbox{SPS}(G,B)$ be a solution. Then there exists a solution $\sigma_0' \in \mbox{SPS}(G,B)$ without cycles such that $\sigma_0' \preceq \sigma_0$, and 
\begin{align}
&\forall c' \in C_{\sigma_0'}, \forall i \in \{1,\ldots,t\}:~~  c'_i \le 2^{\Theta(t^2)}\cdot|V|^{\Theta(t)}\cdot(B+3) ~~\vee~~ c'_i = \infty. \label{eq:O()} 
\end{align}
In case of any general SP game $G \in \DimGame{t}$, the same result holds with $|V|$ replaced by $|V| \cdot W$ in the inequality.
\end{theorem}

In view of this result, a solution to the SPS problem is said to be \emph{Pareto-bounded}  when its Pareto-optimal costs are bounded as stated in the theorem. 

The theorem is proved by induction on the dimension $t$, with the calculation of a \emph{function $\f{B}{t}$ depending on both $B$ and~$t$}, that bounds the components $c'_i \neq \infty$. That is, in Theorem~\ref{thm:bounded-costs}, Equation~(\ref{eq:O()}) is replaced by 
\begin{align}
&\forall c' \in C_{\sigma_0'}, \forall i \in \{1,\ldots,t\}:~~  c'_i \le \f{B}{t} ~~\vee~~ c'_i = \infty. \label{eq:f} 
\end{align}
This function is defined by induction on $t$ through the proofs, and afterwards made explicit and upper bounded by the bound given in Equation (\ref{eq:O()}):
\begin{eqnarray}
 \label{eq:fBounded}
 \f{B}{t} \leq 2^{\Theta(t^2)}\cdot|V|^{\Theta(t)}\cdot(B+3).
\end{eqnarray}

The proof of Theorem~\ref{thm:bounded-costs} is detailed in the next four subsections for binary SP games; it is then easily adapted to any SP games by Lemma~\ref{lem:binary}.

\subsection{Dimension One}
%==========================

We begin the proof of Theorem~\ref{thm:bounded-costs} with the case $t=1$. In this case, the order on costs is total. 

\begin{lemma} \label{lem:basis}
Let $G \in \BinDimGame{1}$ be a binary SP game with dimension $1$, $B \in \nat$, and $\sigma_0 \in \mbox{SPS}(G,B)$ be a solution. Then there exists a solution $\sigma_0' \in \mbox{SPS}(G,B)$ without cycles such that $\sigma_0' \preceq \sigma_0$ and 
\begin{eqnarray}
\forall c' \in C_{\sigma_0'}: ~~c' \le \f{B}{1} = B + |V| ~~\vee~~ c' = \infty.
\end{eqnarray}
\end{lemma} 

Notice that $\f{B}{1}$ satisfies (\ref{eq:fBounded}) when $t=1$. Indeed, we have 
\begin{eqnarray}\label{eq:1}
\f{B}{1} \leq |V|\cdot(B+3) \leq 2^{\Theta(t^2)}\cdot|V|^{\Theta(t)}\cdot(B+3).
\end{eqnarray}

Lemma~\ref{lem:basis} is proved by showing that if the unique Pareto-optimal cost of $\sigma_0$
is finite but greater than $B + |V|$, then we can eliminate a cycle according to Lemma~\ref{lem:cycle} and get a better solution.

\begin{proof}[Proof of Lemma~\ref{lem:basis}] Let $\sigma_0 \in \mbox{SPS}(G,B)$ be a solution. 
Player~$1$ has only one target, thus $C_{\sigma_0}$ is a singleton, say $C_{\sigma_0} = \{c\}$. If $c \le B + |V|$ or $c = \infty$, Lemma~\ref{lem:basis} trivially holds (we eliminate cycles if necessary). Therefore, let us suppose that $B + |V| < c < \infty$. Let $\rho \in \wit_{\sigma_0}$ be such that $\cost{\rho} = c$. Let $h$ be the history of maximal length such that $h \prefix \rho$ and $w(h) = B$. Notice that $h$ exists as the arena is binary and $B + |V| < c < \infty$. Since $\sigma_0 \in \mbox{SPS}(G,B)$ and $\rho \in \wit_{\sigma_0}$, Player~0's target is visited by $h$ and Player~$1$'s target is visited by $\rho$ at least $|V|+1$ vertices after $h$. Hence, $\rho$ performs a cycle between the two visits, that satisfies the hypotheses of Lemma~\ref{lem:cycle}. We can thus eliminate this cycle and create a better solution. We repeat this process until $c \le B + |V|$ and the termination is guaranteed by the strict reduction of $\length{\wit_{\sigma_0}}$ (by Lemma~\ref{lem:cycle}). In this way, Lemma~\ref{lem:basis} is established.
\end{proof}

\subsection{Bounding the Pareto-Optimal Costs by Induction}
%=============================================

We now proceed to the case of dimension $t + 1$, with $t \geq 1$. The next lemma is proved by using the \emph{induction hypothesis} (that is, we suppose that Equation~(\ref{eq:f}) holds for dimension $t$). Recall that $c_{min}$ is the minimum component of the cost~$c$.

\begin{lemma} \label{lem:induction}
Let $G \in \BinDimGame{t+1}$ be a binary SP game with dimension $t+1$, $B \in \nat$, and $\sigma_0 \in \mbox{SPS}(G,B)$ be a solution. Then there exists a solution $\sigma_0' \in \mbox{SPS}(G,B)$ without cycles such that $\sigma_0' \preceq \sigma_0$, and 
\begin{eqnarray} \label{eq:induction}
\forall c' \in C_{\sigma_0'}, \forall i \in \{1,\ldots,t+1\}:  ~~c'_i \le \max\{c'_{min},B\} + 1 + \f{0}{t} ~~\vee~~ c'_i = \infty.
\end{eqnarray}
\end{lemma}

The idea of the proof is as follows. If there exists $c \in C_{\sigma_0}$ such that for some $i \in \{1,\ldots,t+1\}$, $c_i$ does not satisfy (\ref{eq:induction}), then we consider a witness $\rho$ with $\cost{\rho} = c$ and the history $h$ of minimal length such that $h \prefix \rho$ and $w(h) = \max\{c_{min},B\} + 1$. It follows that the subgame $G_{|h}$ has a smaller dimension and we can thus apply the induction hypothesis in $G_{|h}$ with $B = 0$ as $h$ has already visited Player~$0$'s target. Hence, by (\ref{eq:f}), we get a better solution in the subgame $G_{|h}$, and then a better solution in the whole game $G$ by Lemma~\ref{lem:improve}. 

Let us proceed to the formal proof.

\begin{proof}[Proof of Lemma~\ref{lem:induction}]
In this proof, we assume that Equation~(\ref{eq:f}) holds for dimension $t$, by induction hypothesis. Let $\sigma_0 \in \mbox{SPS}(G,B)$ be a solution. If Equation~(\ref{eq:induction}) is satisfied, then Lemma~\ref{lem:induction} holds (we eliminate cycles if necessary).
Therefore, let us suppose that there exists $c \in C_{\sigma_0}$ such that for some $i \in \{1,\ldots,t+1\}$: 
\begin{eqnarray} \label{eq:ci}
\max\{c_{min},B\} + 1 + \f{0}{t}< c_i < \infty.
\end{eqnarray} 
Let $\rho \in \wit_{\sigma_0}$ be a witness with $\cost{\rho} = c$. We define the history $h$ of minimal length such that
\begin{eqnarray} \label{eq:h}
h \prefix \rho \mbox{ and } w(h) = \max\{c_{min},B\} + 1.
\end{eqnarray}
Notice that $h$ exists by definition of $c_i$ and as the arena is binary. As $\sigma_0 \in \mbox{SPS}(G,B)$, by definition of $h$, Player~0's target and at least one target of Player~$1$ are visited by $\bar h$ such that $h = \bar{h}\Last{h}$. Therefore the subgame $G_{|h}$ has dimension $k \leq t$. Notice that $k>0$ by definition of $c_i$, see (\ref{eq:ci}).

Let us consider the SP game $\bar G$ 
\begin{itemize}
\item with the same arena as $G$ but with the initial vertex $\Last{h}$ (instead of $v_0$),
\item with Player~$0$'s target equal to $\{\Last{h}\}$ (since $\bar h$ has visited the target of Player~$0$ in $G$, we replace it by the initial vertex of $\bar G$),\footnote{We recall that a target must be non-empty by definition.}
\item and with the targets visited by $\bar h$ removed from Player~$1$'s set of targets, and if $k < t$ with $t-k$ additional targets that are copies of some remaining (non-empty) targets in a way to have exactly $t$ targets for Player~$1$.
\end{itemize} 
This game $\bar G$ has dimension $t$. We also consider the strategy $\bar{\sigma}_0$ for Player~$0$ in $\bar G$ constructed from the strategy $\sigma_{0|h}$ in $G_{|h}$ as follows: $\bar{\sigma}_0(g) = \sigma_{0|h}(\bar h g)$ for all histories $g \in \hist_{\bar G}$ (Player~$0$ plays in $\bar G$ from $\Last{h}$ as he plays in $G_{|h}$ from $h$). We have that $\bar{\sigma}_0 \in \SPS{\bar{G}}{0}$ (the bound is equal to $0$ as Player~$0$'s target is equal to $\Last{h}$). 

We can thus apply the induction hypothesis: by Equation~(\ref{eq:f}), there exists $\bar{\tau}_0 \in \SPS{\bar{G}}{0}$ without cycles such that $\bar{\tau}_0 \preceq \bar{\sigma}_0$ and for all $\bar c \in C_{\bar{\tau}_0}$ and all $j \in \{1,\ldots,k\}, \bar{c}_j \le \f{0}{t}$ or $\bar{c}_j = \infty$. 

Notice that one can choose for $\bar{\sigma}_0$ a set of witnesses derived from those of $\sigma_0$ having $h$ as prefix: $\wit_{\bar{\sigma}_0} = \{\bar{\pi} \in \play_{\bar G} \mid \bar{h}\bar{\pi} \in \wit_{\sigma_0}\}$. In particular, there exists $\bar \rho \in \wit_{\bar{\sigma}_0}$ such that $\rho = \bar{h}\bar{\rho}$. By (\ref{eq:ci}) and (\ref{eq:h}), we have that $\costT{T_i}{\bar{\rho}} > \f{0}{t}$ ($\costT{T_i}{\bar{\rho}}$ and $\costT{T_i}{\rho}$ differ by $w(h)$). As $\bar{c}_i \leq \f{0}{t}$, it follows that $\bar{\tau}_0 \prec \bar{\sigma}_0$. 

We now want to transfer the previous solution $\bar{\tau}_0$ in $\bar G$ to the subgame $G_{|h}$ in a way to apply Lemma~\ref{lem:improve} and thus obtain the desired strategy $\sigma'_0$ in $G$. Recall that $\bar{\sigma}_0$ was constructed from  $\sigma_{0|h} \in \Sigma_{0|h}$. Let us conversely define the strategy $\tau_0 \in \Sigma_{0|h}$ from $\bar{\tau}_0$: $\tau_0(\bar h g) = \bar{\tau}_0(g)$ for all histories $\bar h g \in \hist_{G}$. Moreover, we can choose the following sets of witnesses for $\sigma_{0|h}$ and $\tau_0$: $\wit_{\sigma_{0|h}} = \{\bar{h}\bar{\pi} \mid \bar{\pi} \in \wit_{\bar{\sigma}_0}\}$ and $\wit_{\tau_0} = \{\bar{h}\bar{\pi} \mid \bar{\pi} \in \wit_{\bar{\tau}_0}\}$. From $\bar{\tau}_0  \prec \bar{\sigma}_0$, it follows  that $\tau_0 \prec \sigma_{0|h}$ (again, the Pareto-optimal costs of $\bar{\tau}$ and $\tau$ differ by $w(h)$, and so do the ones of $\bar{\sigma}_0$ and $\sigma_0$). 
Moreover $\tau_0 \in \SPS{G_{|h}}{B}$ since $h$ visits Player~$0$'s target. Hence, by Lemma~\ref{lem:improve}, the strategy $\sigma'_0 = \sigma_0[h \rightarrow \tau_0]$ is a solution in $\SPS{G}{B}$ and $\sigma_0' \prec \sigma_0$.

We repeat the process described above as long as there remain costs $c \in C_{\sigma_0}$ that are too large. The process terminates as we are always building strictly better strategies. If the resulting strategy is not without cycles, we can eliminate them, one by one, to get a better strategy by applying Lemma~\ref{lem:cycle}. This second process also terminates by the strict reduction of the length of the witness set (by Lemma~\ref{lem:cycle}).
\end{proof}

\subsection{Bounding the Minimum Component of Pareto-Optimal Costs}
%==================================================================

%\subsection{Bounding the Number of Pareto-Optimal Costs}

To prove Theorem~\ref{thm:bounded-costs}, in view of Lemma~\ref{lem:induction}, our last step is to provide a bound on $c_{min}$, the minimum component of each Pareto-optimal cost $c \in C_{\sigma_0}$. Notice that if $c_{min} = \infty$, then all the components of $c$ are equal to $\infty$. In this case, $C_{\sigma_0} = \{(\infty, \ldots, \infty) \}$, i.e., there is no play in $\play_{\sigma_0}$ visiting Player~$1$'s targets. Lemma~\ref{lem:bounded} provides a bound on $c_{min}$  when $C_{\sigma_0} \neq \{(\infty, \ldots, \infty) \}$.

\begin{lemma} \label{lem:bounded}
Let $G \in \BinDimGame{t+1}$ be a binary SP game with dimension $t+1$, $B \in \nat$, and $\sigma_0 \in \mbox{SPS}(G,B)$ be a solution without cycles and satisfying (\ref{eq:induction}). Suppose that $C_{\sigma_0} \neq \{(\infty, \ldots, \infty) \}$. Then, 
\begin{eqnarray}
    \forall c \in C_{\sigma_0}: ~c_{min} &\le& B + 2^{t+1}\big(|V| \cdot (\log_2(|C_{\sigma_0}|)+1) + 1 + \f{0}{t}\big),
    \label{eq:cmin}\\
    \mbox{ with } \quad 
    |C_{\sigma_0}| &\le& \big( \f{0}{t}+B+3 \big)^{t+1}. \label{eq:Csigma0}
\end{eqnarray}
\end{lemma}

%This bound on $c_{min}$ depends on $|C_{\sigma_0}|$, a bound of which is first given in the next lemma.

%\begin{lemma} \label{lem:boundC}
%Let $G \in \BinDimGame{t+1}$ be a binary SP game with dimension $t+1$, $B \in \nat$, and $\sigma_0 \in \mbox{SPS}(G,B)$ be a solution satisfying (\ref{eq:induction}). Suppose that $C_{\sigma_0} \neq \{(\infty, \ldots, \infty) \}$. Then
%\begin{eqnarray} \label{eq:Csigma0}
%    |C_{\sigma_0}| \le \big( \f{0}{t}+B+3 \big)^{t+1}.
%\end{eqnarray}
%\end{lemma}

%\begin{proof}
%Let $\sigma_0$ be a solution such that for all $c \in C_{\sigma_0}$, for all $i \in \{1,\ldots,t+1\}$, $c_i \le \max\{c_{min},B\} + 1 + \f{0}{t}$ or $c_i = \infty$. Therefore, we can write each $c \in C_{\sigma_0}$ as $c = c_{min}(1,\ldots,1) + d$ with $d_i \in \{0, \ldots, B + 1 + \f{0}{t}\} \cup \{\infty\}$ for all $i$. If two costs $c, c' \in C_{\sigma_0}$ are such that $c = c_{min}(1,\ldots,1) + d$ and $c' = c'_{min}(1,\ldots,1) + d$, with the same vector $d$, then they are comparable. This is impossible as  $C_{\sigma_0}$ is an antichain. Hence, the size of $C_{\sigma_0}$ is bounded by the number of vectors $d$, that is, by $\big( \f{0}{t}+B+3 \big)^{t+1}$.
%\end{proof}

The proof of Lemma~\ref{lem:bounded} requires the next technical property about trees. We recall the notion of \emph{depth} of a node in a tree: the root has depth $0$, and if a node has depth $d$, then its sons have depth $d+1$.

\begin{lemma} \label{lem:binarytree}
    Let $n, \ell \in \nat \setminus \{0\}$. Consider a finite tree with at most $n$ leaves such that there are at most $\ell - 1$ consecutive nodes with degree one along any branch of the tree. Then the leaves with the smallest depth have a depth bounded by $\ell \cdot (\log_2(n)+1)$.
\end{lemma}

\begin{proof}
We are going to proceed by induction on $n$. If $n = 1$, then the tree only contains nodes of degree at most 1. Thus, the depth of the tree is at most $l-2 \leq l \cdot (\log_2(n) + 1)$.

Let $n \geq 2$ and suppose that the lemma holds for any $k < n$. Consider $\mathcal{T}$ a finite tree with at most $n$ leaves such that there are at most $l-1$ consecutive nodes with degree one along any branch of the tree. Let $X$ be the set of nodes of $\mathcal{T}$ whose degree is greater than or equal to 2. Since $n \ge 2$, $X$ is not empty. Let $x$ be a node in $X$ with the smallest depth. By definition, $x$ is connected to the root of $\mathcal{T}$ by a succession of nodes of degree 1. Thus, its depth is at most $l-1$. Let $d \ge 2$ be the degree of $x$ and $x_1,\ldots,x_d$ be its children. Since $d \ge 2$, there exists $i$ such that the subtree $\mathcal{T}_i$ rooted in $x_i$ has at most $\lfloor n/2 \rfloor$ leaves. Hence, by induction hypothesis, the depth of the leaves of $\mathcal{T}_i$ of smallest depth is bounded by $l \cdot(\log_2(\lfloor n/2 \rfloor) +1) \leq l \cdot \log_2(n)$. Since the depth of $x_i$ in $\mathcal T$ is $l$, the depth of the leaves of $\mathcal{T}$ of smallest depth is bounded by $l + l \cdot \log_2(n) = l\cdot(\log_2(n) + 1)$.  
\end{proof}

We are now able to establish Lemma~\ref{lem:bounded}.

\begin{proof}[Proof of Lemma~\ref{lem:bounded}]
Let $\sigma_0 \in \mbox{SPS}(G,B)$ be a solution without cycles and satisfying (\ref{eq:induction}). Suppose that $C_{\sigma_0} \neq \{(\infty, \ldots, \infty) \}$ and let $\wit_{\sigma_0}$ be a set of witnesses for $C_{\sigma_0}$. Assume by contradiction that Inequality (\ref{eq:cmin}) does not hold, that is, there exists $d \in C_{\sigma_0}$ such that 
\begin{eqnarray} \label{eq:d}
    B + 2^{t+1}\big(\delta + 1 + \f{0}{t}\big) < d_{min} < \infty. 
\end{eqnarray}
with $\delta = |V| \cdot (\log_2(|C_{\sigma_0}|)+1)$. Let $\rho \in \wit_{\sigma_0}$ be a witness such that $\cost{\rho} = d$. We are going to build a finite sequence of Pareto-optimal costs $(c^{(k)})_{k \in \{0,\ldots,2^{t+1}\}}$ such that, for all $k \in \{0,\ldots,2^{t+1}-1\}$,
\begin{eqnarray} \label{eq:construction}
\max \{ c^{(k)}_i \mid c^{(k)}_i \ne \infty \} < c^{(k+1)}_{min}.    
\end{eqnarray}
Given the size of this sequence, by the pigeonhole principle, there must exist $k,k' \in \{0,\ldots,2^{t+1}\}$ such that $k < k'$ and 
$$\{i \in \{1,\ldots,t+1\} \mid c_i^{(k)} = \infty\} = \{i \in \{1,\ldots,t+1\} \mid c_i^{(k')} = \infty\}$$ 
(a cost component is either finite or infinite). It follows from (\ref{eq:construction}) that $c^{(k)} < c^{(k')}$. This is impossible as $C_{\sigma_0}$ is an antichain.

To build the sequence $(c^{(k)})_{k \in \{0,\ldots,2^{t+1}\}}$, we consider the tree $\mathcal{T}$ of witnesses of $\wit_{\sigma_0}$, obtained by truncating each witness at its first visit to a target of Player~1. Notice that each witness visits at least one target of Player~$1$ as $C_{\sigma_0} \neq \{(\infty, \ldots, \infty) \}$, which explains that this truncation always exists. This truncated tree is finite: $\mathcal{T}$ has at most $|C_{\sigma_0}|$ leaves that correspond to the histories $g$, prefixes of witnesses, such that $w(g) = c_{min}$, with $c \in C_{\sigma_0}$, and such that the first visit of $g$ of some target of Player~$1$ is in its last vertex $\Last{g}$. 
Notice that one leaf of $\mathcal T$ corresponds to some $g$ such that $w(g) = d_{min}$.

Let us construct the first Pareto-optimal cost $c^{(0)}$. Let $h_0$ be the history of maximal length such that $h_0 \prefix \rho$ and $w(h_0) = B$. This history $h_0$ exists because the arena is binary and $\cost{\rho} = d$ with $d_{min}$ satisfying (\ref{eq:d}). Moreover, as $\sigma_0$ is a solution, $h_0$ visits Player~$0$'s target. We consider the subtree $\mathcal{T}_0$ of $\mathcal T$ rooted in the last vertex of $h_0$. This subtree has at most $|C_{\sigma_0}|$ leaves as it is the case for $\mathcal T$. Its nodes with degree $\geq 2$ correspond to histories that are branching points\footnote{We recall that a branching point is a history which is the greatest common prefix of two witnesses.}. Moreover,  any two consecutive nodes with degree one along any branch of ${\mathcal T}_0$ are in the same region, as Player~$0$'s target is visited by $h_0$ and the first visit to a target of Player~$1$ is in a leaf of $\mathcal{T}_0$. As $\sigma_0$ is without cycle in the sense of Lemma~\ref{lem:cycle}, it follows that there are at most $|V|-1$
consecutive nodes with degree one along any branch of $\mathcal{T}_0$. 

Therefore, we can apply Lemma~\ref{lem:binarytree} to $\mathcal{T}_0$ with the parameters $n = |C_{\sigma_0}|$ and $\ell = |V|$. It follows that the leaves of $\mathcal{T}_0$ with the smallest depth have a depth $\leq \delta = |V| \cdot (\log_2(|C_{\sigma_0}|)+1)$. We set $c^{(0)}$ as the cost of a witness associated with one of those leaves. We get that $c^{(0)}_{min} \le B  + \delta$ because $\mathcal{T}_0$ is the subtree rooted at $h_0$ with $w(h_0) = B$. As $\sigma_0$ satisfies (\ref{eq:induction}) by hypothesis, we get that $\max\{c^{(0)}_i \mid c^{(0)}_i \ne \infty\} \le \max\{{c^{(0)}_{min},B}\} + 1 + \f{0}{t}$, that is,
\begin{eqnarray} \label{eq:c0}
 \max\{c^{(0)}_i \mid c^{(0)}_i \ne \infty\} \leq  B + \big(\delta + 1 + \f{0}{t}\big).
\end{eqnarray}

Let us construct the second Pareto-optimal cost $c^{(1)}$ as we did for $c^{(0)}$. Let $h_1$ be the history of maximal length such that $h_1 \prefix \rho$ and $w(h_1) = B + \delta + 1 + \f{0}{t}$ (notice that we use the bound of (\ref{eq:c0})). This history $h_1$ exists for the same reasons as for $h_0$. Let $\mathcal{T}_1$ be the subtree rooted in the last vertex of $h_1$ (it is a subtree of $\mathcal{T}_0$). We can apply Lemma~\ref{lem:binarytree} as for $\mathcal{T}_0$: the leaves of $\mathcal{T}_1$ with the smallest depth have a depth $\leq \delta$ .  Thus, we set $c^{(1)}$ as the cost of a witness associated with one of these leaves. We get that $c^{(1)}_{min} \le B + \delta + 1 + \f{0}{t}  + \delta$ by definition of $w(h_1)$. By (\ref{eq:induction}), we get
\begin{eqnarray}
 \max\{c^{(1)}_i \mid c^{(1)}_i \ne \infty\} \leq  B + 2\big(\delta + 1 + \f{0}{t}\big).
\end{eqnarray}
We also have that $\max \{ c^{(0)}_i \mid c^{(0)}_i \ne \infty \} < c^{(1)}_{min}$ as required in (\ref{eq:construction}). Indeed, $c^{(1)}$ and $d$ correspond to two different leaves of $\mathcal{T}_1$, and thus $c^{(1)}$ does not correspond to the root of $\mathcal{T}_1$. By definition of $h_1$, we get that $c^{(1)}_{min} > w(h_1) = B + \delta + 1 + \f{0}{t}$, and thus $c^{(1)}_{min} > \max \{ c^{(0)}_i \mid c^{(0)}_i \ne \infty \}$ by (\ref{eq:c0}).

As $d_{min}$ satisfies (\ref{eq:d}), we can repeat this process to construct the next costs $c^{(2)}, c^{(3)}, \ldots,$ until the last cost $c^{(2^{t+1})}$. This completes the construction of the announced sequence $(c^{(k)})_{k \in \{0,\ldots,2^{t+1}\}}$.

It remains to prove Inequality~(\ref{eq:Csigma0}). By hypothesis, the solution $\sigma_0$ satisfies (\ref{eq:induction}), that is, for all $c \in C_{\sigma_0}$, for all $i \in \{1,\ldots,t+1\}$, $c_i \le \max\{c_{min},B\} + 1 + \f{0}{t}$ or $c_i = \infty$. Therefore, we can write each $c \in C_{\sigma_0}$ as $c = c_{min}(1,\ldots,1) + d$ with $d_i \in \{0, \ldots, B + 1 + \f{0}{t}\} \cup \{\infty\}$ for all $i$. If two costs $c, c' \in C_{\sigma_0}$ are such that $c = c_{min}(1,\ldots,1) + d$ and $c' = c'_{min}(1,\ldots,1) + d$, with the same vector $d$, then they are comparable. This is impossible as  $C_{\sigma_0}$ is an antichain. Hence, the size of $C_{\sigma_0}$ is bounded by the number of vectors $d$, that is, by $\big( \f{0}{t}+B+3 \big)^{t+1}$, and (\ref{eq:Csigma0}) is then established.
\end{proof}

\subsection{Proof of Theorem~\ref{thm:bounded-costs}}
%========================================================

Finally, we gather all our intermediate results and combine them, in a way to complete the proof of Theorem~\ref{thm:bounded-costs}. Thanks to Lemmas~\ref{lem:basis}-\ref{lem:bounded}, calculations can be done in a way to have an explicit formula for $\f{B}{t}$ and a bound on its value.

\begin{proof}[Proof of Theorem~\ref{thm:bounded-costs}]
Let $G$ be a binary SP game with dimension $t$, $B \in \nat$, and $\sigma_0 \in \mbox{SPS}(G,B)$ be a solution. We assume that $C_{\sigma_0} \neq \{(\infty,\ldots, \infty)\}$, as Theorem~\ref{thm:bounded-costs} is trivially true in case $C_{\sigma_0} = \{(\infty,\ldots, \infty)\}$. 
By Lemmas~\ref{lem:basis}-\ref{lem:bounded}, there exists $\sigma'_0 \in \mbox{SPS}(G,B)$ without cycles such that $\sigma'_0 \preceq \sigma_0$ and 
$$\forall c' \in C_{\sigma'_0}, \forall i \in \{1,\ldots,t\}:~ c'_i \le \f{B}{t} ~\vee~ c'_i = \infty,$$ where
\begin{itemize}
    \item in dimension $t = 1$: $\f{B}{1} = B + |V|,$
    \item in dimension $t+1$ (under the induction hypothesis for $t$): 
        \begin{description}
            \item[(a)] $c'_i \leq \max\{c'_{min},B\} + 1 + \f{0}{t} ~\vee~ c'_i = \infty$,
            \item[(b)] $c'_{min} \le B + 2^{t+1}\big(|V|\cdot (\log_2(|C'_{\sigma_0}|)+1) + 1 + \f{0}{t}\big)$,
            \item[(c)] $|C'_{\sigma_0}| \le \big( \f{0}{t}+B+3 \big)^{t+1}$,
        \end{description}
        that is, 
        
        \centerline{$\begin{array}{ll}
        \f{B}{t+1} = \big(B + 2^{t+1}\big(|V|\cdot (\log_2(\alpha)+1) + 1 + \f{0}{t}\big)\big) + 1 + \f{0}{t}
        & (d)
        \end{array}$}
        
        with $\alpha = \big( \f{0}{t}+B+3 \big)^{t+1}$.
\end{itemize}

We have to prove Inequation (\ref{eq:fBounded}) holds, that is, 
\begin{eqnarray*} 
\f{B}{t} \le 2^{\Theta(t^2)}\cdot |V|^{\Theta(t)}\cdot (B+3)
\end{eqnarray*}
for all $t \geq 1$. This is true when $t=1$ by (\ref{eq:1}). Let us suppose that it is true for $t$, and let us prove that it remains true for $t+1$.
Let us begin with the factor $|V|\cdot(\log_2(\alpha)+1) + 1 + \f{0}{t}$ appearing in the bound on $c'_{min}$:
$$\begin{array}{lll} 
    &|V|\cdot(\log_2(\alpha)+1) + 1 + \f{0}{t}  \\
    &= |V| \cdot \big((t+1) \cdot \log_2(\f{0}{t}+B+3) + 1 \big) + 1 + \f{0}{t} &\mbox{as $\log_2(x^k) =  k \log_2(x)$}\\
    &\leq |V| \cdot \big( (t+1) \cdot (\f{0}{t}+B+3) + 1\big) + 1 + \f{0}{t} &\mbox{as } \log_2(x)\leq x\\
    &\leq (\f{0}{t} + B + 3)\cdot |V|\cdot(t+2) & \mbox{(e).}
\end{array}$$
Let us now compute a bound on $\f{B}{t+1}$:
$$\begin{array}{lll} 
    &\f{B}{t+1} \\
    &\leq B + 2^{t+1}\Big((\f{0}{t} + B + 3)\cdot |V|\cdot (t+2) + 1 + \f{0}{t} \Big) & \mbox{by (d) and (e)}\\
    &\leq 2^{t+1}\big(\f{0}{t} + B + 3\big)\cdot |V| \cdot (t+3) & \mbox{(f).} 
\end{array}$$
It follows that $\f{B}{t+1}$ can be computed thanks to $\f{0}{t}$ bounded by induction:
$$\begin{array}{lll}
   &\f{B}{t+1} \leq \\
    &\leq 2^{t+1}\big(3 \cdot 2^{\Theta(t^2)}\cdot |V|^{\Theta(t)} + B + 3\big)\cdot|V|\cdot(t+3) & \mbox{by (f) and (\ref{eq:O()})}\\
    &\leq 2^{\Theta((t+1)^2)}\cdot |V|^{\Theta(t+1)}\cdot(B+3) &\mbox{as } 3 \cdot t+3 \leq 2^{t+3}.
\end{array}$$
Therefore, Inequation (\ref{eq:fBounded}) holds for $\f{B}{t+1}$. This completes the proof of Theorem~\ref{thm:bounded-costs} by induction on $t$. 
\end{proof}

Thanks to Lemmas~\ref{lem:binary}, \ref{lem:bounded} and Theorem~\ref{thm:bounded-costs}, we easily get a bound for $|C_{\sigma_0}|$ depending on $G$ and $B$, as stated in the next corollary. 

\begin{corollary} \label{cor:Csigma0}
For all games $G \in \DimGame{t}$ and for all Pareto-bounded\footnote{The notion of Pareto-bounded solution has been defined below Theorem~\ref{thm:bounded-costs}.} solutions $\sigma_0 \in \mbox{SPS}(G,B)$, the size $|C_{\sigma_0}|$ is either equal to 1 or bounded exponentially by $2^{\Theta(t^3)}\cdot (|V| \cdot W)^{\Theta(t^2)}\cdot(B+3)^{\Theta(t)}$.
\end{corollary}

\begin{proof}
Suppose that $G$ is binary. Then 
$$|C_{\sigma_0}| \leq \big( \f{0}{t}+B+3 \big)^{t+1}$$
by Lemma~\ref{lem:bounded} and
$$\f{B}{t} \le 2^{\Theta(t^2)}\cdot |V|^{\Theta(t)}\cdot(B+3)$$
by (\ref{eq:fBounded}). Therefore, we have
$$\begin{array}{ll}
   |C_{\sigma_0}| &\leq \big( 2^{\Theta(t^2)}\cdot |V|^{\Theta(t)}\cdot 3 +B+3 \big)^{t+1} \\
   &\leq \big( 2^{\Theta(t^2)+2}\cdot |V|^{\Theta(t)} \cdot (B+3) \big)^{t+1} \\
   &\leq 2^{\Theta(t^3)}\cdot |V|^{\Theta(t^2)}\cdot(B+3)^{\Theta(t)}.
   \end{array}$$
In the general case, by Lemma~\ref{lem:binary}, $|V|$ is multiplied by $W$ in the last inequality. 
\end{proof}

In the sequel, we use the same notation $\f{B}{t}$ for any SP games, having in mind that $|V|$ has to be multiplied by $W$ when the game arena is not binary.

\section{Complexity of the SPS Problem} \label{sec:complexity}
%=======================================

In this section, we prove our main result: the SPS problem is \textsf{NEXPTIME}-complete (Theorem~\ref{thm:main}). It follows the same pattern as for Boolean reachability~\cite{DBLP:conf/concur/BruyereRT21}, however it requires the results of Section~\ref{sec:bounding} (which is meaningless in the Boolean case) and some modifications to handle quantitative reachability. 

\subsection{Finite-Memory Solutions}
%==================================

We first show that if there exists a solution to the SPS problem, then there is one that is finite-memory and whose memory size is bounded exponentially. This intermediate step is necessary to prove that the SPS problem is in \textsf{NEXPTIME}: we will guess such a strategy and check that it is a solution in exponential time.

\begin{proposition} \label{prop:finiteMemory}
Let $G$ be an SP game, $B \in \nat$, and $\sigma_0 \in \mbox{SPS}(G,B)$ be a solution. Then there exists a Pareto-bounded solution $\sigma_0' \in \mbox{SPS}(G,B)$ such that $\sigma_0'$ is finite-memory and its memory size is bounded exponentially.
\end{proposition}

When $C_{\sigma_0} \neq \{(\infty,\ldots, \infty)\}$, the proof of this proposition is based on the following principles, that are detailed below (the case $C_{\sigma_0} = \{(\infty,\ldots, \infty)\}$ is treated separately):

\begin{itemize}
\item We first transform the arena of $G$ into a binary arena and adapt the given solution $\sigma_0 \in \mbox{SPS}(G,B)$ to the new game. We keep the same notations $G$ and $\sigma_0$. We can suppose that $\sigma_0$ is Pareto-bounded by Theorem~\ref{thm:bounded-costs}. We consider a set of witnesses $\wit_{\sigma_0}$.
\item We show that at any deviation\footnote{We recall that a deviation is a history $hv$ with $h \in \histi{1}$, $v \in V$, such that $h$ is prefix of some witness, but $hv$ is prefix of no witness.}, Player~$0$ can switch to a \emph{punishing strategy} that imposes that the consistent plays $\pi$ either satisfy $\val{\pi} \leq B$ or $\cost{\pi}$ is not Pareto-optimal. %, i.e., it gives to Player~1 a cost that is strictly greater than the cost of a witness. 
Moreover, this punishing strategy is finite-memory with an exponential memory.
\item We then show how to transform the witnesses into lassos, and how they can be produced by a finite-memory strategy with exponential memory. We also show that we need at most exponentially many different punishing strategies.
\end{itemize}
In this way, we get a strategy solution to the SPS problem whose memory size is exponential.\medskip

\subsection{Punishing Strategies}
%=================================

Let $\sigma_0$ be a Pareto-bounded solution to the SPS Problem. By Theorem~\ref{thm:bounded-costs} (\ref{eq:f}) (\ref{eq:fBounded}), we get that $c_i \leq \f{B}{t}$ or $c_i = \infty$ for all $c \in C_{\sigma_0}$ and all $i \in \{1,\ldots,t\}$, such that $\f{B}{t}$ is exponentially bounded. Moreover, if a play $\rho$ is Pareto-optimal, then $\val{\rho} \leq B$. We define for each history $g \in \hist_{\sigma_0}$ its \emph{record} $\mem{h} =$ $(w(h),\val{h}, \cost{h})$ whose values are \emph{truncated} to $\infty$ if they are greater than $\f{B}{t}$. 

 We define for each deviation $hv$ a punishing strategy $\punStrat{v,\mem{hv}}$ as stated in the next lemma. %Notice that such a strategy is associated with a \emph{small} history, such that $|h| \leq V \cdot W \cdot B$. This bound will be clarified later.

\begin{lemma} \label{lem:punish}
Let $G$ be an SP game, $B \in \nat$, and $\sigma_0 \in \SPS{G}{B}$ be a Pareto-bounded solution. Suppose $C_{\sigma_0} \neq \{(\infty,\ldots, \infty)\}$. Let $hv$ be a deviation such that $\val{h} = \infty$ (resp. $\val{h} < \infty$). Then there exists a finite-memory strategy $\punStrat{v,\mem{hv}}$ with an exponential memory size (resp. with size~$1$) such that $\sigma_0' = \sigma_0[hv \rightarrow \punStrat{v,\mem{hv}}]$ is also a solution in $\mbox{SPS}(G,B)$.
\end{lemma}

\begin{proof}
Let $hv$ be a deviation such that $\val{h} = \infty$, that is, $h$ does not visit Player~$0$'s target. As $\sigma_0$ is a solution, notice that $\sigma_{0|hv}$ imposes to each consistent play $\pi$ in $G_{|hv}$ to satisfy $\val{\pi} \leq B$ or $\cost{\pi} > c$ for some $c \in C_{\sigma_0}$. We are going to replace $\sigma_{0|hv}$ by a winning strategy in a zero-sum game $H$ with an exponential arena and an omega-regular objective\footnote{We assume that the reader is familiar with the concept of zero-sum game with an omega-regular objective, see e.g. \cite{2001automata}.} that is equivalent to what $\sigma_{0|hv}$ imposes to plays. The arena of $H$ is the arena of $G$ extended with the record $\mem{g}$ of the current history $g$. More precisely, vertices are of the form $(v,(m_1,m_2,m_3))$ with $v \in V$, $m_1,m_2 \in \{0,\ldots, \f{B}{t}\} \cup \{\infty\}$, and $m_3 \in \big(\{0,\ldots, \f{B}{t}\} \cup \{\infty\}\big)^t$, such that, whenever $v$ belongs to some target, the weight component $m_1$ allows to update the (truncated) $\textsf{val}$ component $m_2$ and the (truncated) $\textsf{cost}$ component $m_3$. The initial vertex of $H$ is equal to $(v,\mem{hv})$. This arena is finite and of exponential size by the way the values of $(m_1,m_2,m_3)$ are truncated and by Theorem~\ref{thm:bounded-costs} (\ref{eq:fBounded}). The omega regular objective of $H$ is the disjunction between
\begin{itemize}
\item %\emph{(1)} 
the reachability objective $\{(v,(m_1,m_2,m_3)) \mid m_2 \leq B\}$, and \item %\emph{(2)} 
the safety objective $\{(v,(m_1,m_2,m_3)) \mid m_3 > c $ for some $c \in C_{\sigma_0}\}$.
\end{itemize}
It is known, see e.g.~\cite{BruyereHR18}, that if there exists a winning strategy for zero-sum games with an objective which is the disjunction of a reachability objective and a safety objective, then there exists one that is memoryless. This is the case here for the extended game $H$: as $\sigma_{0|hv}$ is winning in $H$, there exists a winning memoryless strategy in $H$, and thus a winning finite-memory strategy $\punStrat{v,\mem{hv}}$ with exponential size in the original game. Therefore, the strategy $\sigma_0[hv \rightarrow \punStrat{v,\mem{hv}}]$ is again a solution in $\mbox{SPS}(G,B)$.

Suppose now that $hv$ is a deviation with $\val{h} < \infty$. As $h$ already visits Player~$0$'s target, we can use, in place of $\sigma_{0|hv}$, any memoryless strategy $\punStrat{v,\mem{hv}}$ as deviating strategy.
\end{proof}

\subsection{Lasso Witnesses}
%===========================

To get Proposition~\ref{prop:finiteMemory}, we show that one can play with an exponential finite-memory strategy over the witness tree and an exponential number of deviating strategies. The main idea of the proof is to transform each witness $\rho$ into a \emph{lasso}: as the solution $\sigma_0$ is Pareto-bounded, in the region decomposition of $\rho$, each region contains no cycle, except the last one; in the last region, as soon as a vertex is repeated, we replace the suffix of $\rho$ by the infinite repetition of this cycle. We will show that the number of regions traversed by these lasso witnesses is exponential.

\begin{proof}[Proof of Proposition~\ref{prop:finiteMemory}]
(1) We first suppose that $C_{\sigma_0} \neq \{(\infty,\ldots, \infty)\}$. According to Theorem~\ref{thm:bounded-costs}, from $\sigma_0 \in \mbox{SPS}(G,B)$, one can construct a solution $\sigma_0' \in \mbox{SPS}(G,B)$ that is Pareto-bounded, thus without cycles. Consider a set of witnesses $\wit_{\sigma'_0}$  for this strategy, and the region decomposition of its witnesses. Along any witness $\rho$, according to Lemma~\ref{lem:cycle}, it is thus impossible to eliminate cycles. Recall that Lemma~\ref{lem:cycle} only considers cycles implying histories in the same region, those cycles with a null weight before visiting Player~$0$'s target; moreover, the last region of $\rho$ is excluded. Let us study the memory used by $\sigma'_0$. Along a witness $\rho$:
\begin{itemize}
    \item As $\val{\rho} \leq B$ and there is no cycle with a null weight before visiting Player~$0$'s target, it follows that the smallest prefix $h$ of $\rho$ such that $\val{h} = \val{\rho}$ has a length $|h| \leq |V|\cdot B \cdot W$ (the presence of $W$ comes from the fact that the arena has been made binary).
    \item Once Player~$0$'s target is visited and before the last region, $\sigma'_0$ is ``locally'' memoryless inside each region, as there is no cycle. 
    \item In the last region, as soon as a vertex is repeated, we replace $\sigma'_0$ by a memoryless strategy that repeats this cycle forever. We then get a lasso replacing $\rho$ that has the same value and cost. We keep the same notation $\rho$ for this lasso. We also keep the notation $\wit_{\sigma'_0}$ for the set of lasso witnesses. Notice that for deviations $hv$ such that $h$ belongs to the last region of a witness, the strategy $\punStrat{v,\mem{hv}}$ is memoryless (as $\val{h} < \infty$, see Lemma~\ref{lem:punish}). Hence the modification of the witnesses has no impact on the deviating strategies.
\end{itemize}

Let us study the memory necessary for $\sigma'_0$ in order to produce the new set $\wit_{\sigma'_0}$: 
\begin{itemize}
\item The number of regions traversed by a witness $\rho$ is bounded by $(t+2) \cdot |\wit_{\sigma'_0}|$. Indeed, the number of visited targets increases from $0$ to $t+1$ and the set $\wit_{\sigma'_0}(h)$ decreases until being equal to $\{\rho\}$. As there are $|C_{\sigma'_0}|$ witnesses, the total number of traversed regions is bounded by $(t+2) \cdot |C_{\sigma'_0}|^2$ which is exponential by Corollary~\ref{cor:Csigma0}. 
\item On the first regions traversed by a witness $\rho$, $\sigma'_0$ has to memorize the current history $h \prefix \rho$ until $|h| = |V|\cdot B \cdot W$, which needs an exponential memory. Then, on each other region traversed by $\rho$, $\sigma'_0$ is locally memoryless. Therefore, the memory necessary for $\sigma'_0$ to produce the lasso witnesses is exponential, as there is an exponential number of regions on which $\sigma'_0$ needs an exponential memory.
\end{itemize}

It remains to explain how to play with a finite-memory strategy in case of deviations from the witnesses. With each history $g$, we associate its record $\mem{g}$ whose second component indicates whether $\val{g} < \infty$ or $\val{g} = \infty$. Let $hv$ be a deviation. We apply Lemma~\ref{lem:punish} to replace $\sigma'_{0|hv}$ by the punishing strategy $\punStrat{v,\mem{hv}}$. According to $\val{h}$, this punishing strategy is either finite-memory with an exponential size or it is memoryless. Notice that given a deviation $hv$, a punishing strategy only depends on the last vertex $v$ of $hv$ and its record $\mem{hv}$. Therefore, we have an exponential number of punishing strategies by Theorem~\ref{thm:bounded-costs}.
 
All in all, when $C_{\sigma_0} \neq \{(\infty,\ldots, \infty)\}$, we get from $\sigma_0$ a solution to the SPS problem that uses a memory of exponential size.

\bigskip
(2) Suppose that $C_{\sigma_0} = \{(\infty,\ldots, \infty)\}$. This means that each play $\rho \in \play_{\sigma_0}$ visits no target of Player~$1$ and that $\sigma_0$ can impose to each such $\rho$ to have a value $\val{\rho} \leq B$. As done in the proof of Lemma~\ref{lem:punish}, we can replace $\sigma_{0}$ by a strategy corresponding to a winning strategy in a zero-sum game $H$ with an exponential arena and a reachability objective. The arena of $H$ has vertices of the form $(v,(m_1,m_2))$ with $v \in V$, $m_1,m_2 \in \{0,\ldots, B\} \cup \{\infty\}$, such that, whenever $v$ belongs to Player~$0$'s target, the weight component $m_1$ allows to update the \textsf{val} component $m_2$. The initial vertex of $H$ is equal to $(v_0,0,\val{v_0})$. This arena has an exponential size and its reachability objective is $\{(v,(m_1,m_2)) \mid m_2 \leq B\}$. It is known that when there is a winning strategy in a zero-sum game with a reachability objective, then there is one that is memoryless~\cite{2001automata}. Hence, coming back to $G$, $\sigma_0$ can be replaced by a finite-memory strategy with exponential memory.
\end{proof}

\subsection{\textsf{NEXPTIME}-Completeness}
%==========================================

We finally prove that the SPS problem is \textsf{NEXPTIME}-complete. For the \textsf{NEXPTIME}-membership, the idea is to guess a strategy $\sigma_0$ given by a Mealy machine of exponential size (by Proposition~\ref{prop:finiteMemory}) and then to verify in exponential time that it is a solution. The \textsf{NEXPTIME}-hardness follows from the \textsf{NEXPTIME}-completeness of the Boolean variant of the SPS problem~\cite{DBLP:conf/concur/BruyereRT21}.

\begin{proof}[Proof of Theorem~\ref{thm:main}]
Let us first study the \textsf{NEXPTIME}-membership. Let $G$ be an SP game and $B\in \nat$. If there exists a solution, Proposition~\ref{prop:finiteMemory} states the existence of a solution $\sigma_0 \in \SPS{G}{B}$ that uses a finite memory bounded exponentially. %Without loss of generality, we can assume that the arena of $G$ is binary (thus with a size $|V| \cdot W$). 
We can guess such a strategy $\sigma_0$ as a Mealy machine $\mathcal{M}$ with a set $M$ of memory states at most exponential in the size of the instance. This can be done in exponential time. Let us explain how to verify in exponential time that the guessed strategy $\sigma_0$ is a solution to the SPS problem.%, i.e., that every play $\rho$ in $\play_{\sigma_0}$ whose cost is Pareto-optimal has a value $\val{\rho} \leq B$.

We make the following \emph{important observation}: it is enough to consider Pareto-optimal costs $c$ whose components $c_i < \infty$ belong to $\{0, \ldots, |V| \cdot |M| \cdot t \cdot W\}$. Let us explain why: 
\begin{itemize}
    \item Consider the cartesian product $G \times \mathcal{M}$ whose infinite paths are exactly the plays consistent with $\sigma_0$. This product has an arena of size $|V| \cdot |M|$ where Player~$1$ is the only player to play. The Pareto-optimal costs are among the costs of plays $\rho$ in $G \times \mathcal{M}$ that have no cycle with positive weight between two consecutive visits of Player~$1$'s targets\footnote{or between the initial vertex and the first visit to some Player~$1$'s target.}.
    \item Consider now a Pareto-optimal play $\rho$ with a cycle of null weight between two consecutive visits of Player~$1$'s targets. Then there exists another play $\rho'$, with the same cost as $\rho$, that is obtained by removing this cycle. Moreover, if $\val{\rho} = \infty$, then $\val{\rho'} = \infty$. 
    \item {Therefore, it is enough to consider plays $\rho$ such that for $h \prefix \rho$ with length $|h| = |V| \cdot |M| \cdot t$, we have $\cost{\rho} = \cost{h}$ (the worst case happens when $\rho$ visit all Player~$1$'s targets each of them separated by a longest path without cycle). It follows that $\cost{\rho} \leq |V| \cdot |M| \cdot t \cdot W$.}
\end{itemize}

From the previous observation, let us explain how to compute the set $C_{\sigma_0}$ of Pareto-optimal costs, and then how to check that $\sigma_0$ is a solution.
\begin{itemize}
    \item First, we further extend the vertices of $G \times \mathcal{M}$ to keep track of the weight, value and cost of the current history, truncated to $\infty$ when they are greater than $\alpha = \max\{B,|V| \cdot |M| \cdot t \cdot W\}$). That is, we consider an arena $H$ whose vertices are of the form $(v,s,(m_1,m_2,m_3))$ with $v \in V$, $s \in M$, $m_1,m_2 \in \{0,\ldots,\alpha\} \cup \{\infty\}$, and $m_3 \in (\{0,\ldots,\alpha\} \cup \{\infty\})^t$. As in the proof of Lemma~\ref{lem:punish}, whenever $v$ belongs to some target, the weight component $m_1$ allows to update the \textsf{val} component $m_2$ and the \textsf{cost} component $m_3$. The initial vertex is $(v_0,s_0,(0,\val{v_0},\cost{v_0}))$ where $s_0$ is the initial memory state of $\mathcal M$. %The extended arena $H$ has exponential size. 
    \item Then, to compute $C_{\sigma_0}$, we test for the existence of a play $\rho \in \play_{\sigma_0}$ with a given cost $c = \cost{\rho}\in (\{0,\ldots,\alpha\} \cup \{\infty\})^t$, beginning with the smallest possible cost $c = (0,\ldots,0)$, and finishing with the largest possible one $c = (\infty,\ldots,\infty)$. Deciding the existence of such a play $\rho$ for cost $c$ corresponds to deciding the existence of a play in the extended arena $H$ that visits some vertex $(v,s,(m_1,m_2,m_3))$ with $m_3 = c$. This corresponds to a reachability objective that can be checked in polynomial time in the size of $H$, thus in exponential time in the size of the given instance. Therefore, as there is at most an exponential number of costs $c$ to consider, the set $C_{\sigma_0}$ can be computed in exponential time.
    \item Finally, we check whether $\sigma_0$ is \emph{not a solution}, i.e., there exists a play $\rho$ in $H$ with a cost $c \in C_{\sigma_0}$ such that $\val{\rho} > B$.  We %simply 
    remove from $H$ all vertices $(v,s,(m_1,m_2,m_3))$ such that $m_2 \leq B$, and we then check the existence of a play with a cost $c \in C_{\sigma_0}$ as done in the previous item. As this can be done in exponential time, checking that $\sigma_0$ is a solution can thus be done in exponential time. 
\end{itemize}

Let us finally study the \textsf{NEXPTIME}-hardness of the SPS problem. In~\cite{DBLP:conf/concur/BruyereRT21}, the Boolean variant of the SPS problem is proved to be \textsf{NEXPTIME}-complete. It can be reduced to its quantitative variant by labeling each edge with a weight equal to $0$ and by considering a bound $B$ equal to $0$. Hence the \textsf{val} and \textsf{cost} components are either equal to $0$ or $\infty$. It follows that the (quantitative) SPS problem is \textsf{NEXPTIME}-hard.
\end{proof}

\section{Conclusion and Future Work}
%===================================

In \cite{DBLP:conf/concur/BruyereRT21}, the SPS problem is proved to be \textsf{NEXPTIME}-complete for Boolean reachability. In this paper, we proved that the same result holds for quantitative reachability (with non-negative weights). The difficult part was to show that when there exists a solution to the SPS problem, there is one whose Pareto-optimal costs are exponentially bounded.

In this paper, we suppose that there is one weight function $w : E \rightarrow \mathbb{N}$ used to compute the value and the cost of any play $\rho$, i.e. to compute $\costT{T_0}{\rho}$ and $(\costT{T_1}{\rho}, \ldots, \costT{T_t}{\rho})$. We could have considered the more general case where there are $t+1$ weight functions $w_i$, $i \in \{0,1, \ldots t\}$, and each cost $\costT{T_i}{\rho})$ is computed thanks to $w_i$. Unfortunately, the techniques developed in this paper cannot be adapted to this case. Indeed, Lemma~\ref{lem:cycle} allows to eliminate cycles, with the restriction that the cycle has a null weight as long as Player~$0$'s target $T_0$ is not visited yet. It follows that in the proof of Proposition~\ref{prop:finiteMemory} where we transform each witness into a lasso by eliminating cycles, we can guarantee that this lasso visits $T_0$ along its prefix $h$ with length $|V|\cdot B\cdot W$. In the case of games with $t+1$ (dissociated) weight functions $w_i$, we do not see how to bound the length of such a prefix $h$.

This seems like a challenging problem to solve, in view of some very recent results obtained in~\cite{Christophe}. The \emph{rational synthesis} problem is investigated in~\cite{Christophe} such that the environment is composed of several players whose rational responses are to play an NE (instead of one player with several targets following a play with a Pareto-optimal cost). Both concepts of rationally coincide in the following case: on one hand, an environment composed of one player playing an NE, on the other hand an environment composed of one player with one target following a play with a Pareto-optimal cost. The reference~\cite{Christophe} uses several weight functions, one for each player, and the obtained results are: (1) when the environment is composed of one player, the rational synthesis problem is \textsf{PSPACE}-hard and in \textsf{EXPTIME},
(2) this problem is in general \textsf{EXPTIME}-hard and it is unknown whether it is decidable. For the SPS problem with several weight functions, we thus have the following results: (1) when Player~$1$ has one target, the SPS problem is \textsf{PSPACE}-hard and in \textsf{EXPTIME}, (2) this problem is in general \textsf{NEXPTIME}-hard (as a consequence of our Theorem~\ref{thm:main}.

In~\cite{Christophe}, weight functions with positive and negative weights have been investigated: it is proved that in both cases of rationality (NE/Pareto-optimality), the synthesis problem becomes undecidable as soon as negative weights are allowed. The proof requires to use several weight functions. It would be interesting to study whether this problem remains undecidable in case of one weight function $w$.  

Considering multiple objectives for Player~$0$ (instead of one) is also an interesting problem to study. The order on the tuples of values becomes partial. In this case, we could impose a threshold on each value component.

It is well-known that quantitative objectives make it possible to model richer properties than with Boolean objectives. This paper studied quantitative reachability. It would be very interesting to investigate the SPS problem for other quantitative payoffs, like mean-payoff or discounted sum.

\bibliography{biblio}

\end{document}